\documentclass[journal,draftcls,onecolumn,12pt,twoside]{IEEEtran}
\usepackage{amsfonts,color,morefloats,pslatex}
\usepackage{amssymb, amsthm, amsmath,latexsym}

\usepackage{longtable}
\newtheorem{theorem}{Theorem}[section]
\newtheorem{lemma}[theorem]{Lemma}

\newtheorem{corollary}[theorem]{Corollary}
\newtheorem{conj}[theorem]{Conjecture}
\newtheorem{example}{Example}

\allowdisplaybreaks

\usepackage[para]{threeparttable}

\newcommand{\tr}{{\mathrm{Tr}}}

\newcommand{\Norm}{{\mathrm{Norm}}}
\newcommand{\gf}{{\mathrm{GF}}}
\newcommand{\PG}{{\mathrm{PG}}}

\newcommand{\wt}{{\mathtt{wt}}}

\newcommand{\F}{{\mathbb{F}}}

\newcommand{\bC}{{\mathbb{C}}}

\newcommand{\C}{{\mathcal{C}}}

\newcommand{\bc}{{\mathbf{c}}}

\makeatletter

\newcommand{\Rmnum}[1]{\expandafter\@slowromancap\romannumeral #1@}
\makeatother
\newcommand\myatop[2]{\genfrac{}{}{0pt}{}{#1}{#2}}

\ifCLASSINFOpdf

\else

\fi

\begin{document}
%
\title{The Subfield Codes of $[q+1, 2, q]$ MDS Codes\thanks{Z. Heng's research was supported in part by the National Natural Science Foundation of China under grant number 11901049, in part by the the Natural Science Basic Research Program of Shaanxi under grant number 2020JQ-343, and in part by the Young Talent fund of University Association for Science and Technology in Shaanxi, China.
C. Ding's research was supported by the Hong Kong Research Grants Council,
Proj. No. 16300919.}
}

\author{
Ziling Heng\thanks{Z. Heng is with the School of Science, Chang'an University, Xi'an 710064, China (email: zilingheng@163.com)},
 Cunsheng Ding\thanks{C. Ding is with the Department of Computer Science and Engineering, The Hong Kong University of Science and Technology, Clear Water Bay, Kowloon, Hong Kong, China (email: cding@ust.hk)}
 }

\maketitle

\begin{abstract}
Recently, subfield codes of geometric codes over large finite fields $\gf(q)$ with dimension $3$ and $4$ were studied and distance-optimal subfield codes over $\gf(p)$ were obtained, where $q=p^m$. The key idea for obtaining very good subfield codes over small fields is to choose very good linear codes over an extension field with small dimension. This paper first presents a general construction of $[q+1, 2, q]$ MDS codes over $\gf(q)$, and then studies the subfield codes over $\gf(p)$ of some of the $[q+1, 2,q]$ MDS codes over $\gf(q)$. Two families of dimension-optimal codes over $\gf(p)$ are obtained, and several families of nearly optimal codes over $\gf(p)$ are produced. Several open problems are also proposed in this paper.

\end{abstract}


%
\IEEEpeerreviewmaketitle

\section{Introduction}

 Let $q$ be a prime power. Let $n, k, d$ be positive integers.
 An $[n,\, k,\, d]$ \emph{code} $\C$ over $\gf(q)$ is a $k$-dimensional subspace of $\gf(q)^n$ with minimum
 (Hamming) distance $d$.
 Let $A_i$ denote the number of codewords with Hamming weight $i$ in a code
 $\C$ of length $n$. The {\em weight enumerator} of $\C$ is defined by
 $1+A_1z+A_2z^2+ \cdots + A_nz^n$, and
 the sequence $(1, A_1, A_2, \cdots, A_n)$ is called the \emph{weight distribution} of the code $\C$.
 A code $\C$ is said to be a $t$-weight code  if the number of nonzero
 $A_i$ in the sequence $(A_1, A_2, \cdots, A_n)$ is equal to $t$.
 An $[n,\, k,\, d]$ code over $\gf(q)$ is called \emph{distance-optimal} if there is no
 $[n,\, k,\, d']$ code over $\gf(q)$ with $d'>d$,  \emph{dimension-optimal} if there is no
 $[n,\, k',\, d]$ code over $\gf(q)$ with $k' >k$, and  \emph{length-optimal} if there is no
 $[n',\, k,\, d]$ code over $\gf(q)$ with $n' <n$.  A code is said to be optimal if it is
 distance-optimal, or dimension-optimal, or length-optimal, or it meets a bound for
 linear codes. An $[n, k, d]$ code over $\gf(q)$  is said to be \emph{nearly optimal} if an $[n-1, k, d]$ code over $\gf(q)$ is
 distance-optimal, or an $[n, k+1, d]$ code over $\gf(q)$ is dimension-optimal, or an $[n, k, d+1]$ code over
 $\gf(q)$ is distance-optimal.

\subsection{Subfield codes and their properties}\label{sec-subfieldcodes}

Let $\gf(p^m)$ be a finite field with $p^m$ elements, where $p$ is a prime and $m$ is a positive integer. In this section, we introduce subfield codes of linear codes and some basic results of subfield codes.

Given an $[n,k]$ code $\C$ over $\gf(p^m)$, we construct a new $[n, k']$ code $\C^{(p)}$
over $\gf(p)$ as follows. Let $G$ be a generator matrix of $\C$. Take a basis of $\gf(p^m)$
over $\gf(p)$. Represent each entry of $G$ as an $m \times 1$ column vector of $\gf(p)^m$
with respect to this basis, and replace each entry of $G$ with the corresponding $m \times 1$ column vector of $\gf(p)^m$. In this way, $G$ is modified into a $km \times n$ matrix over
$\gf(p)$, which generates the new subfield code $\C^{(p)}$ over $\gf(p)$ with length $n$.
By definition, the dimension $k'$ of $\C^{(p)}$ satisfies $k'\leq mk$. It was proved in
\cite{DH19} that
the subfield code $\C^{(p)}$ of $\C$ is independent of the choices of both $G$ and the
basis of $\gf(p^m)$ over $\gf(p)$. Note that the subfield code $\C^{(p)}$ contains the
subfield subcodes over $\gf(p)$ of $\C$ as a subset and the two codes over $\gf(q)$
are different in general. Notice that the subfield subcodes have been well studied \cite{Dels}.

The following two theorems document basic properties of subfield codes \cite{DH19}.

\begin{theorem}\label{th-tracerepresentation}
Let $\C$ be an $[n,k]$ linear code over $\gf(p^m)$. Let $G=[g_{ij}]_{1\leq i \leq k, 1\leq j \leq n}$ be a generator matrix of $\C$. Then the trace representation of $\C^{(p)}$ is given by
$$
\C^{(p)}=\left\{\left(\tr_{p^m/p}\left(\sum_{i=1}^{k}a_ig_{i1}\right),
\cdots,\tr_{p^m/p}\left(\sum_{i=1}^{k}a_ig_{in}\right)\right):a_1,\ldots,a_k\in \gf(p^m)\right\}.
$$
\end{theorem}

Denote by $\C^{\perp}$ and $\C^{(p)\perp}$ the dual codes of $\C$ and its subfield code $\C^{(p)}$, respectively. Let $\C^{\perp (p)}$ denote the subfield code of $\C^{\perp}$.
Since the dimensions of $\C^{\perp (p)}$ and $\C^{(p)\perp}$ vary from case to case,
there may not be a general relation between the two codes $\C^{\perp (p)}$ and $\C^{(p)\perp}$.

A relation between the minimal distance of $\C^{\perp}$ and that of $\C^{(p)\perp}$ is given as follows.
\begin{theorem}\label{th-dualdistance}
Let $\C$ be an $[n,k]$ linear code over $\gf(p^m)$. Then the minimal distance $d^\perp$ of $\C^{\perp}$ and the minimal distance $d^{(p)\perp}$ of $\C^{(p)\perp}$ satisfy
$d^{(p)\perp}\geq d^\perp.$
\end{theorem}

Two linear codes $\C_1$ and $\C_2$ are {\em permutation equivalent\index{permutation equivalent
of codes}} if there is a permutation of coordinates which sends $\C_1$ to $\C_2$. If $\C_1$ and $\C_2$
are permutation equivalent, so are $\C_1^\perp$ and $\C_2^\perp$. Two permutation equivalent
linear codes have the same dimension and weight distribution.

A \emph{monomial matrix\index{monomial matrix}} over a field $\F$ is a square matrix having exactly one
nonzero element of $\F$  in each row and column. A monomial matrix $M$ can be written either in
the form $DP$ or the form $PD_1$, where $D$ and $D_1$ are diagonal matrices and $P$ is a permutation
matrix.

Let $\C_1$ and $\C_2$ be two linear codes of the same length over $\F$. Then $\C_1$ and $\C_2$
are \emph{monomially equivalent\index{monomially equivalent}} if there is a monomial matrix over $\F$
such that $\C_2=\C_1M$. Monomial equivalence and permutation equivalence are precisely the same for
binary codes. If $\C_1$ and $\C_2$ are monomially equivalent, then they have the same weight distribution.

Let $\C$ and $\C'$ be two monomially equivalent $[n,k]$ code over $\gf(p^m)$.
Let $G=[g_{ij}]$ and $G=[g'_{ij}]$ be two generator matrices of $\C$ and $\C'$,
respectively. By definition, there exist a permutation $\sigma$ of the set
$\{1, 2, \cdots, n\}$ and elements $b_1, b_2, \cdots, b_n$ in $\gf(p^m)^*$
such that
$$
g_{ij}=b_jg'_{i\sigma(j)}
$$
for all $1 \leq i \leq k$ and $1 \leq j \leq n$. It then follows that
\begin{eqnarray*}
\lefteqn{\left(\tr_{p^m/p}\left(\sum_{i=1}^{k}a_ig_{i1}\right),
\cdots,\tr_{p^m/p}\left(\sum_{i=1}^{k}a_ig_{in}\right)\right) } \\
&=&
\left(\tr_{p^m/p}\left(b_1\left(\sum_{i=1}^{k}a_ig'_{i\sigma(1)}\right)\right),
\cdots,\tr_{p^m/p}\left(b_n\left(\sum_{i=1}^{k}a_ig'_{i\sigma{n}}\right)\right)\right).
\end{eqnarray*}
Then the following conclusions follow from Theorem \ref{th-tracerepresentation}:
\begin{itemize}
\item If $\C$ and $\C'$ are permutation equivalent, so are $\C^{(p)}$ and $\C'^{(p)}$.
\item If all $b_i \in \gf(p)^*$, then $\C^{(p)}$ and $\C'^{(p)}$ are monomially equivalent.
\end{itemize}
However, $\C^{(p)}$ and $\C'^{(p)}$ may not be monomially equivalent even if $\C$ and $\C'$
are monomially equivalent.

\subsection{The motivations and objectives of this paper}

Every linear code $\C$ over $\gf(p^m)$ has a subfield code  $\C^{(p)}$, which may have very good or bad parameters.
To obtain a subfield code $\C^{(p)}$ with desirable parameters, one has to select the code $\C$ over $\gf(p^m)$ properly.
Even if a code $\C$ over $\gf(p^m)$ looks simple, the subfield code  $\C^{(p)}$ could be very complex and it could be very difficult to determine the minimum distance of the subfield code  $\C^{(p)}$, let alone the weight distribution of the subfield code. For example, the Simplex code  $\C$ over $\gf(p^m)$ is extremely simple and a one-weight code, but it is in general very hard to determine the weight distribution of the subfield code  $\C^{(p)}$.

To be able to obtain a very good code $\C^{(p)}$ and settle its parameters, one should select an optimal code or very good code
$\C$ over $\gf(p^m)$ with small dimension. This idea was successfully used in several references for obtaining infinite families of distance-optimal codes.  Linear codes over $\gf(q)$ with parameters $[q^2+1, 4, q^2-q]$ are called \emph{ovoid codes} and optimal with respect to the Griesmer bound, as they correspond to ovoids in $\PG(3, \gf(q))$ \cite[Chapter 13]{Dingbk19}. The subfield codes of some ovoid codes are very good \cite{DH19}.
Linear code over $\gf(2^m)$ with parameters $[2^m+2, 3, 2^m]$ are called \emph{hyperoval codes}, as they correspond to hyperovals in $\PG(2, \gf(2^m))$ \cite[Chapter 12]{Dingbk19}.
Linear code over $\gf(2^m)$ with parameters $[2^m+1, 3, 2^m-1]$ are called \emph{oval codes}, as they correspond to ovals in $\PG(2, \gf(2^m))$ \cite[Chapter 12]{Dingbk19}. The subfield codes of some hyperoval and oval codes are distance-optimal \cite{HD19}. Later, the subfield codes of some hyperoval codes were extended to more general cases \cite{WZ}. Maximal arcs in
$\PG(2, \gf(2^m))$ give $[n, 3, n-h]$ two-weight codes over $\gf(2^m)$, where $h=2^s$ with $1 \leq s <m$ and $n=h2^m+h-2^m$, whose duals have parameters $[n, n-3, 4]$ if $s=1$ and $[n, n-3, 3]$ if $s>1$ \cite[Chapter 12]{Dingbk19}. The subfield codes of some maximal arc codes are distance-optimal \cite{HDW20}. In \cite{HWD}, the subfield codes of some cyclic codes and linear codes with dimension 4 were also investigated.

Motivated by the distance-optimal subfield codes obtained in \cite{DH19,HD19,HDW20, HWD}, in this paper we first present a general construction of $[q+1, 2, q]$ MDS codes over $\gf(q)$, and then study the subfield codes of some of them. Our objective is to construct infinite families of codes over $\gf(p)$ with very good parameters. Some conjectures and open problems are proposed.

\section{Auxiliary results}\label{sec-pre}

In this section, we recall characters and some character sums over finite fields which will be needed in later sections.

Let $p$ be a prime and $q=p^m$. Let $\gf(q)$ be the finite field with $q$ elements and $\alpha$ a primitive element of $\gf(q)$. Let $\tr_{q/p}$ denote the trace function from $\gf(q)$ to $\gf(p)$ given by
$$\tr_{q/p}(x)=\sum_{i=0}^{m-1}x^{p^{i}},\ x\in \gf(q).$$ Denote by $\zeta_p$ a primitive $p$-th root of complex unity.
An \emph{additive character} of $\gf(q)$ is a function $\chi: (\gf(q),+)\rightarrow \bC^{*}$ such that
$$\chi(x+y)=\chi(x)\chi(y),\ x,y\in \gf(q),$$ where $\bC^{*}$ denotes the set of all nonzero complex numbers. For any $a\in \gf(q)$, the function
$$\chi_{a}(x)=\zeta_{p}^{\tr_{q/p}(ax)},\ x\in \gf(q),$$ defines an additive character of $\gf(q)$. In addition, $\{\chi_{a}:a\in \gf(q)\}$ is a group consisting of all the additive characters of $\gf(q)$. If $a=0$, we have $\chi_0(x)=1$ for all $x\in \gf(q)$ and $\chi_0$ is referred to as the trivial additive character of $\gf(q)$. If $a=1$, we call $\chi_1$ the canonical additive character of $\gf(q)$. Clearly, $\chi_a(x)=\chi_1(ax)$.
The orthogonality  relation of additive characters is given by
$$\sum_{x\in \gf(q)}\chi_1(ax)=\left\{
\begin{array}{rl}
q    &   \mbox{ for }a=0,\\
0    &   \mbox{ for }a\in \gf(q)^*.
\end{array} \right. $$

Let $\gf(q)^*=\gf(q)\setminus \{0\}$. A \emph{character} $\psi$ of the multiplicative group $\gf(q)^*$ is a function from  $\gf(q)^*$  to $\bC^{*}$ such that $\psi(xy)=\psi(x)\psi(y)$ for all $(x,y)\in \gf(q)\times \gf(q)$. Define the multiplication of two characters $\psi,\psi'$ by $(\psi\psi')(x)=\psi(x)\psi'(x)$ for $x\in \gf(q)^*$. All the characters of $\gf(q)^*$ are given by
$$\psi_{j}(\alpha^k)=\zeta_{q-1}^{jk}\mbox{ for }k=0,1,\cdots,q-1,$$
where $0\leq j \leq q-2$. Then all these $\psi_j$, $0\leq j \leq q-2$, form a group under the multiplication of characters and are called \emph{multiplicative characters} of $\gf(q)$. In particular, $\psi_0$ is called the trivial multiplicative character and for odd $q$, $\eta:=\psi_{(q-1)/2}$ is referred to as the quadratic multiplicative character of  $\gf(q)$. The orthogonality relation of multiplicative characters is given by
$$\sum_{x\in \gf(q)^*}\psi_j(x)=\left\{
\begin{array}{rl}
q-1    &   \mbox{ for }j=0,\\
0    &   \mbox{ for }j\neq 0.
\end{array} \right. $$

For an additive character $\chi$ and a multiplicative character $\psi$ of $\gf(q)$, the \emph{Gauss sum} $G(\psi, \chi)$ over $\gf(q)$ is defined by
$$G(\psi,\chi)=\sum_{x\in \gf(q)^*}\psi(x)\chi(x).$$
We call $G(\eta,\chi)$ the quadratic Gauss sum over $\gf(q)$ for nontrivial $\chi$. The value of the quadratic Gauss sum is known as follows.

\begin{lemma}\label{quadGuasssum}\cite[Th. 5.15]{LN83}
Let $q=p^m$ with $p$ odd. Let $\chi$ be the canonical additive character of $\gf(q)$. Then
\begin{eqnarray*}G(\eta,\chi)&=&(-1)^{m-1}(\sqrt{-1})^{(\frac{p-1}{2})^2m}\sqrt{q}\\
 &=&\left\{
\begin{array}{lll}
(-1)^{m-1}\sqrt{q}    &   \mbox{ for }p\equiv 1\pmod{4},\\
(-1)^{m-1}(\sqrt{-1})^{m}\sqrt{q}    &   \mbox{ for }p\equiv 3\pmod{4}.
\end{array} \right. \end{eqnarray*}
\end{lemma}

Let $\chi$ be a nontrivial additive character of $\gf(q)$ and let $f\in \gf(q)[x]$ be a polynomial of positive degree. The character sums of the form
$$\sum_{c\in \gf(q)}\chi(f(c))$$ are referred to as \emph{Weil sums}. The problem of evaluating
such character sums explicitly is very difficult in general. In certain special cases, Weil sums can be treated (see \cite[Section 4 in Chapter 5]{LN83}). If $f$ is a quadratic polynomial, the Weil sum has an interesting relationship with quadratic Gauss sums, which is described in the
following lemma.

\begin{lemma}\label{lem-charactersum}\cite[Th. 5.33]{LN83}
Let $\chi$ be a nontrivial additive character of $\gf(q)$ with $q$ odd, and let $f(x)=a_2x^2+a_1x+a_0\in \gf(q)[x]$ with $a_2\neq 0$. Then
$$\sum_{c\in \gf(q)}\chi(f(c))=\chi(a_0-a_1^2(4a_2)^{-1})\eta(a_2)G(\eta,\chi).$$
\end{lemma}

If $f$ is a quadratic polynomial with $q$ even, the Weil sums are evaluated explicitly as follows.

\begin{lemma}\label{lem-charactersum-evenq}\cite[Cor. 5.35]{LN83}
Let $\chi_b$ be a nontrivial additive character of $\gf(q)$ with $b\in \gf(q)^*$, and let $f(x)=a_2x^2+a_1x+a_0\in \gf(q)[x]$ with $q$ even. Then
$$\sum_{c\in \gf(q)}\chi_b(f(c))=\left\{\begin{array}{ll}
\chi_b(a_0)q    &   \mbox{ if }a_2=ba_{1}^{2},\\
0    &   \mbox{ otherwise. }
\end{array} \right.$$
\end{lemma}

The Weil sums can also be evaluated explicitly in the case that $f$ is an affine $p$-polynomial over $\gf(q)$.

\begin{lemma}\label{lem-p-polynomial}\cite[Th. 5.34]{LN83}
Let $q=p^m$ and let
$$f(x)=a_rx^{p^r}+a_{r-1}x^{p^{r-1}}+\cdots+a_1x^{p}+a_0x+a$$ be an affine $p$-polynomial over $\gf(q)$. Let $\chi_b$ be a nontrivial additive character of $\gf(q)$ with $b\in \gf(q)^*$. Then
$$\sum_{c\in \gf(q)}\chi_b(f(c))=\left\{\begin{array}{ll}
\chi_b(a)q    &   \mbox{ if }ba_r+b^pa_{r-1}^p+\cdots+b^{p^{r-1}}a_{1}^{p^{r-1}}+b^{p^{r}}a_{0}^{p^{r}}=0,\\
0    &   \mbox{ otherwise. }
\end{array} \right.$$
\end{lemma}

\section{Properties of $[q+1, 2, q]$ MDS codes over $\gf(q)$}

The weight distribution of any $[n, k, n-k+1]$ MDS code over $\gf(q)$ is known and documented in the following theorem \cite[p. 321]{MS77}.

\begin{theorem}\label{thm-wteMDScode}
Let $\C$ be an $[n,\kappa]$ code over $\gf(q)$ with $d=n-\kappa+1$, and let the weight enumerator of $\C$ be
$1+\sum_{i=d}^{n} A_iz^i$. Then
$$
A_i =\binom{n}{i} (q-1) \sum_{j=0}^{i-d} (-1)^j  \binom{i-1}{j} q^{i-j-d} \mbox{ for all } d \leq i \leq n.
$$
\end{theorem}

It follows from Theorem  \ref{thm-wteMDScode} that any $[q+1, 2, q]$ MDS code over $\gf(q)$ is a one-weight code
with weight enumerator
$
1+(q^2-1)z^q.
$
The dual code of any $[q+1, 2, q]$ MDS code over $\gf(q)$ has parameters $[q+1, q-1, 3]$.
It is well known that every one-weight code is monomially equivalent to the concatenation of several Simplex codes.
Consequently, every  $[q+1, 2, q]$ MDS code over $\gf(q)$ must be monomially equivalent to the Simplex code
with parameters $[q+1, 2, q]$. Hence, up to monomial equivalence, there is only one one-weight code over $\gf(q)$
with parameters $[q+1, 2, q]$. However, the subfield codes of $[q+1, 2, q]$ MDS codes over $\gf(q)$ are very
different, as they have different parameters and weight distributions. Thus, it is still very interesting to investigate
$[q+1, 2, q]$ MDS codes over $\gf(q)$, as they give different subfield codes. This will be demonstrated in subsequent sections.

\section{A general construction of $[q+1, 2, q]$ MDS codes over $\gf(q)$}

Let $q=p^m$ as before, where $m$ is a positive integer and $p$ is a prime. Let $f(x)$ be a polynomial over $\gf(q)$. Define
a $2 \times (q+1)$ matrix over $\gf(q)$ by
\begin{eqnarray}\label{matrixG}
G_{(f,q)}=
\left[
\begin{array}{cccccc}
f(\alpha^0)   & f(\alpha^1) & \cdots & f(\alpha^{q-2}) & 0 & 1 \\
\alpha^0   & \alpha^1 & \cdots & \alpha^{q-2} & 1 & 0
\end{array}
\right],
\end{eqnarray}
where $\alpha$ is a generator of $\gf(q)^*$. Let $\C_{(f, q)}$ denote the linear code over $\gf(q)$ with
generator matrix $G_{(f,q)}$.

\begin{theorem}\label{thm-mainpolycode}
$\C_{(f, q)}$ is a $[q+1, 2, q]$ MDS code if and only if
\begin{enumerate}
\item $f(x) \neq 0$ for all $x \in \gf(q)^*$, and
\item $yf(x)-xf(y) \neq 0$ for all distinct $x$ and $y$ in $\gf(q)^*$.
\end{enumerate}
\end{theorem}

\begin{proof}
Notice that the last two columns of $G_{(f,q)}$ form a submatrix of rank $2$. The code $\C_{(f, q)}$ has dimension $2$
for any polynomial $f$ over $\gf(q)$, and its dual has dimension $q-1$. By the Singleton bound, the minimum distance
$d^\perp$ of the dual code  $\C_{(f, q)}$ is at most $3$.

Assume that Conditions 1) and 2) in this theorem are satisfied. We now prove that $d^\perp =3$. Clearly, no column of
$G_{(f,q)}$ is the zero vector. As a result, we deduce that $d^\perp \geq 2$.
Note that
\begin{eqnarray*}
\left|
\begin{array}{cc}
0 & 1 \\
1 & 0
\end{array}
\right|
=-1.
\end{eqnarray*}
For each $a \in \gf(q)^*$, by Condition 1) we have
\begin{eqnarray*}
\left|
\begin{array}{cc}
f(a) & 0  \\
a &  1
\end{array}
\right|
=f(a) \neq 0.
\end{eqnarray*}
For each $a \in \gf(q)^*$,  we have
\begin{eqnarray*}
\left|
\begin{array}{cc}
f(a) & 1  \\
a & 0
\end{array}
\right|
= a \neq 0.
\end{eqnarray*}
For any pair of distinct elements $a$ and $b$ in $\gf(q)^*$, by Condition 2) we have
\begin{eqnarray*}
\left|
\begin{array}{cc}
f(a) & f(b)  \\
a & b
\end{array}
\right|
= bf(a)-af(b) \neq 0.
\end{eqnarray*}
We then arrived at the conclusion that any two columns of $G_{(f,q)}$ are linearly independent over $\gf(q)$.
Consequently, $d^\perp=3$ and $\C_{(f, q)}^\perp$ is a $[q+1, q-1, 3]$ MDS code over $\gf(q)$. Hence,
$\C_{(f, q)}$ is a $[q+1, 2, q]$ MDS code.

Assume that $\C_{(f, q)}$ is a $[q+1, 2, q]$ MDS code. Then $\C_{(f, q)}^\perp$ is a $[q+1, q-1, 3]$ MDS code over $\gf(q)$.
Suppose that $f(a)=0$ for some $a \in \gf(q)^*$. Then the two columns $(0,1)^T$ and $(f(a), a)^T$ in $G_{(f,q)}$ would be
linearly dependent over $\gf(q)$ and thus $\C_{(f, q)}^\perp$ would have a codeword of Hamming weight $2$. This would be
contrary to the fact that $d^\perp=3$. Hence, Condition 1) must be satisfied.
Suppose that $bf(a)-af(b) =0$ for a pair of distinct elements $a$ and $b$ in $\gf(q)^*$.
Then the two columns $(f(a),a)^T$ and $(f(b), b)^T$ in $G_{(f,q)}$ would be
linearly dependent over $\gf(q)$ and thus $\C_{(f, q)}^\perp$ would have a codeword of Hamming weight $2$. This would be
contrary to the fact that $d^\perp=3$. Hence, Condition 2) must be satisfied.  This completes the proof.
\end{proof}

\begin{theorem}\label{thm-tracerepresent}
The trace representation of the $p$-ary subfield code $\C_{(f, q)}^{(p)}$ of $\C_{(f, q)}$ in Theorem \ref{thm-mainpolycode} is given by
$$\C_{(f,q)}^{(p)}=\left\{{\bc_{(f,q)}}^{(p)}=\left(\left(\tr_{q/p}(af(x)+bx)\right)_{x\in \gf(q)^*},\tr_{q/p}(a),\tr_{q/p}(b)\right):a,b\in \gf(q)\right\}.$$
\end{theorem}
\begin{proof}
The desired conclusion follows from Equation (\ref{matrixG}) and Theorem \ref{th-tracerepresentation}.
\end{proof}

If $p=2$, the weight distribution of the binary subfield code $\C_{(f, q)}^{(2)}$ can be depicted by the  Walsh spectrum of $f(x)$. For a function $f(x)$ from $\gf(q)$ to $\gf(q)$ with $f(0)=0$, its Walsh transform is defined as
$$W_{f}(a,b)=\sum_{x\in \gf(q)}(-1)^{\tr_{q/2}(af(x)+bx)},\ a,b\in \gf(q).$$
\begin{corollary}
Let $p=2$ and $f(x)$ is a function from $\gf(q)$ to $\gf(q)$ with $f(0)=0$. For any codeword ${\bc_{(f,q)}}^{(2)}\in \C_{(f,q)}^{(2)}$ in Theorem \ref{thm-tracerepresent}, its Hamming weight
\begin{eqnarray*}
\wt({\bc_{(f,q)}}^{(2)})=\left\{
\begin{array}{lll}
\frac{q}{2}-\frac{1}{2}W_{f}(a,b)    &   \mbox{ if }\tr_{q/2}(a)=\tr_{q/2}(b)=0,\\
\frac{q}{2}-\frac{1}{2}W_{f}(a,b)+1    &   \substack{\mbox{ if }\tr_{q/2}(a)=0,\tr_{q/2}(b)\neq0 \mbox{ or }\\ \tr_{q/2}(a)\neq0,\tr_{q/2}(b)=0,}\\
\frac{q}{2}-\frac{1}{2}W_{f}(a,b)+2    &   \mbox{ if }\tr_{q/2}(a)\neq 0,\tr_{q/2}(b)\neq 0.\\
\end{array} \right.
\end{eqnarray*}
\end{corollary}

\begin{proof}
By the orthogonality relation of additive characters,
\begin{eqnarray*}
\lefteqn{ \sharp \{x\in \gf(q)^*:\tr_{q/2}(af(x)+bx)\neq 0\} }\\
&= q-1-\frac{1}{2}\sum_{x\in \gf(q)^*}\sum_{y\in \gf(2)}(-1)^{y\left(\tr_{q/2}(af(x)+bx)\right)}\\
&= q-1-\left(\frac{q-1}{2}+\frac{1}{2}\sum_{x\in \gf(q)^*}(-1)^{\tr_{q/2}(af(x)+bx)}\right)\\
&= \frac{q}{2}-\frac{1}{2}W_{f}(a,b).
\end{eqnarray*}
For any codeword ${\bc_{(f,q)}}^{(2)}=\left(\left(\tr_{q/p}(af(x)+bx)\right)_{x\in \gf(q)^*},\tr_{q/2}(a),\tr_{q/2}(b)\right)\in \C_{(f,q)}^{(2)}$ by Theorem \ref{thm-tracerepresent}, we directly obtain its Hamming weight
\begin{eqnarray*}
\wt({\bc_{(f,q)}}^{(2)})=\left\{
\begin{array}{lll}
\frac{q}{2}-\frac{1}{2}W_{f}(a,b)    &   \mbox{ if }\tr_{q/2}(a)=\tr_{q/2}(b)=0,\\
\frac{q}{2}-\frac{1}{2}W_{f}(a,b)+1    &   \substack{\mbox{ if }\tr_{q/2}(a)=0,\tr_{q/2}(b)\neq0 \mbox{ or }\\ \tr_{q/2}(a)\neq0,\tr_{q/2}(b)=0,}\\
\frac{q}{2}-\frac{1}{2}W_{f}(a,b)+2    &   \mbox{ if }\tr_{q/2}(a)\neq 0,\tr_{q/2}(b)\neq 0.\\
\end{array} \right.
\end{eqnarray*}
\end{proof}

In order to obtain $p$-ary subfield code $\C_{(f, q)}^{(p)}$  with good parameters, one should   properly select $f$. In the following sections, we investigate the parameters of
 $\C_{(f, q)}^{(p)}$  with some special polynomials $f$ satisfying the two conditions in Theorem \ref{thm-mainpolycode}.
 As will be seen, there exist many infinite families of polynomials $f$ satisfying the conditions.

\section{The subfield code of $\C_{(f, q)}$  when $f(x)=1$}

Let $f(x)=1$. It is easily checked that the two conditions in Theorem \ref{thm-mainpolycode} are satisfied. Hence the code $\C_{(1, q)}$
is a $[q+1, 2, q]$ MDS code by Theorem \ref{thm-mainpolycode}. By Theorem \ref{thm-tracerepresent},
the $p$-ary subfield code of $\C_{(1, q)}$ is given by
$$\C_{(1,q)}^{(p)}=\left\{{\bc_{(1,q)}}^{(p)}=\left(\left(a+\tr_{q/p}(bx)\right)_{x\in \gf(q)^*},a,\tr_{q/p}(b)\right):\substack{a\in \gf(p)\\b\in \gf(q)}\right\}.$$

 \begin{theorem}\label{th-f=1}
Let $m \geq 2$. Then $\C_{(1,q)}^{(p)}$ is a nearly optimal $[p^m+1, m+1, (p-1)p^{m-1}]$ code with respect to the Griesmer bound, and has  weight enumerator
$$1+p(p^{m-1}-1)z^{(p-1)p^{m-1}}+p^m(p-1)z^{(p-1)p^{m-1}+1}+(p-1)z^{p^m}.$$
The dual code $(\C_{(1,q)}^{(p)})^\perp$ has parameters $[p^m+1, p^m-m, 3]$ and is dimension-optimal with respect to the sphere-packing bound.
\end{theorem}

\begin{proof}
Let $\bc_{(1,q)}^{(p)}=\left(\left(a+\tr_{q/p}(bx)\right)_{x\in \gf(q)^*},a,\tr_{q/p}(b)\right)$ be any codeword in $\C_{(1,q)}^{(p)}$. By the orthogonality relation of additive characters,
\begin{eqnarray*}
\lefteqn{ \sharp \{x\in \gf(q)^*:a+\tr_{q/p}(bx)\neq 0\} } \\
&=& q-1-\sharp \{x\in \gf(q)^*:a+\tr_{q/p}(bx)=0\}\\
&=& q-1-\frac{1}{p}\sum_{x\in \gf(q)^*}\sum_{y\in \gf(p)}\zeta_{p}^{y(a+\tr_{q/p}(bx))}\\
&=& \frac{(p-1)(q-1)}{p}-\frac{1}{p}\sum_{y\in \gf(p)^*}\zeta_{p}^{ya}\sum_{x\in \gf(q)^*}\chi(ybx)\\
&=& \left\{
\begin{array}{lll}
\frac{(p-1)(q-1)}{p}-\frac{q-1}{p}\sum_{y\in \gf(p)^*}\zeta_{p}^{ya}    &   \mbox{ if }b=0\\
\frac{(p-1)(q-1)}{p}+\frac{1}{p}\sum_{y\in \gf(p)^*}\zeta_{p}^{ya}    &   \mbox{ if }b\neq 0\\
\end{array} \right.\\
&=& \left\{
\begin{array}{lll}
0    &   \mbox{ if }a=b=0,\\
p^m-1 & \mbox{ if }a\neq0,b=0,\\
(p-1)p^{m-1}    &   \mbox{ if }a=0,b\neq 0,\\
(p-1)p^{m-1}-1    &   \mbox{ if }a\neq 0,b\neq 0.
\end{array} \right.
\end{eqnarray*}
By definition, we then have
\begin{eqnarray*}
\wt(\bc_{(1,q)}^{(p)})=\left\{
\begin{array}{lll}
0    &   \mbox{ if }a=b=0,\\
p^m & \mbox{ if }a\neq0,b=0,\\
(p-1)p^{m-1}    &   \mbox{ if }a=0,b\neq 0,\tr_{q/p}(b)=0,\\
(p-1)p^{m-1}+1    &   \mbox{ if }a=0,\tr_{q/p}(b)\neq 0,\\
(p-1)p^{m-1}    &   \mbox{ if }a\neq 0,b\neq 0,\tr_{q/p}(b)=0,\\
(p-1)p^{m-1} +1   &   \mbox{ if }a\neq 0,\tr_{q/p}(b)\neq0.
\end{array} \right.
\end{eqnarray*}
Then the weight distribution of $\C_{(1,q)}^{(p)}$ follows. The dimension is $m+1$ as $\wt(\bc_{(1,q)}^{(p)})=0$ if and only if $a=b=0$. Then $\C_{(1,q)}^{(p)}$ has parameters $[p^m+1, m+1, (p-1)p^{m-1}]$. Note that
$$\sum_{i=0}^{m}\left\lceil\frac{(p-1)p^{m-1}}{p^i}\right\rceil=p^m.$$ Hence $\C_{(1,q)}^{(p)}$ is nearly optimal with respect to the Griesmer bound.

By Theorem \ref{th-dualdistance}, the minimal distance $d^{(p)\perp}$ of $\C_{(1, q)}^{(p)\perp}$ satisfies $d^{(p)\perp}\geq 3$ as the dual of $\C_{(1, q)}$ has minimal distance 3.  From the weight distribution of $\C_{(1,q)}^{(p)}$ and the first four Pless power moments in \cite[Page 131]{HP}, we can prove that $A_3^{(p)\perp}>0$, where $A_3^{(p)\perp}$ denotes the number of the codewords with weight 3 in $\C_{(1, q)}^{(p)\perp}$. Then the parameters of $\C_{(1, q)}^{(p)\perp}$ follow. By the sphere-packing bound, one can deduce that  the dual of $\C_{(1, q)}^{(p)}$ is dimension-optimal.
\end{proof}

The following example shows  that $\C_{(1, q)}^{(p)}$ is very attractive.

\begin{example} Let  $\C_{(1,q)}^{(p)}$  be the code in Theorem \ref{th-f=1}.
\begin{enumerate}
\item Let $p=2$ and $m=2$. Then the set $\C_{(1, q)}^{(p)}$ is a $[5,3,2]$ binary code and its dual is a $[5,2,3]$ binary code. Both codes have the best known parameters according to the Code Tables at http://www.codetables.de/. Note that $\C_{(1, q)}^{(p)}$ is a near MDS code in this case.
\item Let $p=2$ and $m=3$. Then the set $\C_{(1, q)}^{(p)}$ is a $[9,4,4]$ binary code and its dual is a $[9,5,3]$ binary code.  Both codes have best known parameters according to the Code Tables at http://www.codetables.de/.
\item Let $p=3$ and $m=2$. Then the set $\C_{(1, q)}^{(p)}$ is a $[10,3,6]$ ternary code and its dual is a $[10,7,3]$ almost MDS ternary code. Both codes have the best known parameters according to the Code Tables at http://www.codetables.de/.
\end{enumerate}
\end{example}

\section{The subfield code of $\C_{(f, q)}$ when $f$ is a monomial}
In this section, let $f$ be a monomial over $\gf(q)$ with $q=p^m$, i.e. $f(x)=x^t$ with $t$ a positive integer.
\begin{lemma}\label{lem-monomial}
Let $f(x)=x^t$ with $t$ a positive integer. Then $\C_{(f, q)}$ is a $[q+1, 2, q]$ MDS code if and only if
$\gcd(q-1,t-1)=1$.
\end{lemma}
\begin{proof}
Due to Theorem \ref{thm-mainpolycode}, the code $\C_{(f, q)}$
is a $[q+1, 2, q]$ MDS code if and only if $z^{t-1} \neq 1$ for all $z$ in $\gf(q)^*\setminus \{1\}$, which is true if and only if
$\gcd(q-1, t-1)=1$. This completes the proof.
\end{proof}

There  obviously  exist infinitely many monomials satisfying the condition in Lemma \ref{lem-monomial}. In the following, we select some special monomials $f$ and investigate the parameters of $\C_{(f, q)}^{(p)}$ and its dual.

\subsection{The subfield code $\C_{(f,q)}^{(p)}$ for $f(x)=x^{p^l+1}$ and $m=2l$}
Let $f(x)=x^{p^l+1}$ and $m=2l$.  Since $\gcd(q-1,t-1)=\gcd(p^{m}-1,p^l)=1$, $\C_{(f, q)}$ is a $[q+1, 2, q]$ MDS code by Lemma \ref{lem-monomial}. By Theorem \ref{thm-tracerepresent},
the $p$-ary subfield code of $\C_{(f, q)}$ is given by
$$\C_{(f,q)}^{(p)}=\left\{{\bc_{(f,q)}}^{(p)}=\left(\left(\tr_{p^l/p}(ax^{p^l+1})+\tr_{q/p}(bx)\right)_{x\in \gf(q)^*},\tr_{p^l/p}(a),\tr_{q/p}(b)\right):\substack{a\in \gf(p^l)\\b\in \gf(q)}\right\}.$$
Before investigating the weight distribution of $\C_{(f, q)}$, we present a few lemmas.

\begin{lemma}\label{lem-N1}
Let $m=2l$ with $l\geq 2$. Denote by \begin{eqnarray*}
N_1=\sharp \left\{(a,b)\in \gf(p^l)^*\times \gf(q):\tr_{p^l/p}\left(\frac{b^{p^l+1}}{a}\right)=0,\tr_{p^l/p}(a)=0\mbox{ and }\tr_{q/p}(b)=0\right\}.
\end{eqnarray*}
Then $$N_1=(p^{l-1}-1)(p^{2l-2}-p^l+p^{l-1}).$$
\end{lemma}

\begin{proof} By definition,
\begin{eqnarray}\label{number-1}
\nonumber N_1&=&\sharp \left\{(a,b)\in \gf(p^l)^*\times \gf(q):\tr_{p^l/p}\left(\frac{b^{p^l+1}}{a}\right)=0,\tr_{p^l/p}(a)=0\mbox{ and }\tr_{q/p}(b)=0\right\}\\
\nonumber&=&\sharp \left\{(a,b)\in \gf(p^l)^*\times \gf(q)^*:\tr_{p^l/p}\left(\frac{b^{p^l+1}}{a}\right)=0,\tr_{p^l/p}(a)=0\mbox{ and }\tr_{q/p}(b)=0\right\}\\
& &+(p^{l-1}-1).
\end{eqnarray}
Denote by $\chi$ and $\phi$ the canonical additive characters of $\gf(p^l)$ and $\gf(q)$, respectively. For fixed $a\in \gf(p^l)^*$ satisfying $\tr_{p^l/p}(a)=0$,
\begin{eqnarray}\label{number-1-1}
\nonumber
\lefteqn{ \sharp \left\{b\in \gf(q)^*:\tr_{p^l/p}\left(\frac{b^{p^l+1}}{a}\right)=0\mbox{ and }\tr_{q/p}(b)=0\right\} } \\
\nonumber&=&\frac{1}{p^2}\sum_{b\in \gf(q)^*}\sum_{y\in \gf(p)}\zeta_{p}^{y\tr_{p^l/p}\left(\frac{b^{p^l+1}}{a}\right)}\sum_{z\in \gf(p)}\zeta_{p}^{z\tr_{q/p}(b)}\\
\nonumber &=&\frac{1}{p^2}\sum_{b\in \gf(q)^*}\sum_{y\in \gf(p)}\chi\left(\frac{yb^{p^l+1}}{a}\right)\sum_{z\in \gf(p)}\phi(bz)\\
\nonumber &=&\frac{q-1}{p^2}+\frac{1}{p^2}\sum_{b\in \gf(q)^*}\sum_{y\in \gf(p)^*}\chi\left(\frac{yb^{p^l+1}}{a}\right)+\frac{1}{p^2}\sum_{b\in \gf(q)^*}\sum_{z\in \gf(p)^*}\phi(bz)\\
 & &+\frac{1}{p^2}\sum_{b\in \gf(q)^*}\sum_{y\in \gf(p)^*}\chi\left(\frac{yb^{p^l+1}}{a}\right)\sum_{z\in \gf(p)^*}\phi(bz).
\end{eqnarray}
Since $\Norm(b)=b^{p^l+1}$ is the norm function from $\gf(q)^*$ onto $\gf(p^l)^*$,
\begin{eqnarray*}
\sum_{b\in \gf(q)^*}\sum_{y\in \gf(p)^*}\chi\left(\frac{yb^{p^l+1}}{a}\right)&=&\sum_{y\in \gf(p)^*}\sum_{b\in \gf(q)^*}\chi\left(\frac{yb^{p^l+1}}{a}\right)\\
&=&(p^l+1)\sum_{y\in \gf(p)^*}\sum_{b'\in \gf(p^l)^*}\chi\left(\frac{yb'}{a}\right)\\
&=&-(p^l+1)(p-1).
\end{eqnarray*}
By the orthogonality relation of additive characters,
 $$\sum_{b\in \gf(q)^*}\sum_{z\in \gf(p)^*}\phi(bz)=\sum_{z\in \gf(p)^*}\sum_{b\in \gf(q)^*}\phi(bz)=-(p-1).$$
 Therefore, Equation (\ref{number-1-1}) yields
 \begin{eqnarray}\label{number-1-2}
\nonumber
\lefteqn{ \sharp \left\{b\in \gf(q)^*:\tr_{p^l/p}\left(\frac{b^{p^l+1}}{a}\right)=0\mbox{ and }\tr_{q/p}(b)=0\right\} } \\
&=&\frac{q-1-(p-1)(p^l+2)}{p^2}+\frac{1}{p^2}\sum_{b\in \gf(q)^*}\sum_{y\in \gf(p)^*}\chi\left(\frac{yb^{p^l+1}}{a}\right)\sum_{z\in \gf(p)^*}\phi(bz).
\end{eqnarray}
Note that
\begin{eqnarray*}
\lefteqn{ \sum_{b\in \gf(q)^*}\sum_{y\in \gf(p)^*}\chi\left(\frac{yb^{p^l+1}}{a}\right)\sum_{z\in \gf(p)^*}\phi(bz) } \\
&=&\sum_{y,z\in \gf(p)^*}\sum_{b\in \gf(q)^*}\chi\left(\frac{yb^{p^l+1}}{a}\right)\phi(bz)\\
&=&\sum_{y,z\in \gf(p)^*}\sum_{b\in \gf(q)^*}\chi\left(\frac{yb^{p^l+1}}{a}+bz+b^{p^l}z\right)\\
&=&\sum_{y,z\in \gf(p)^*}\chi\left(-\frac{az^2}{y}\right)\sum_{b\in \gf(q)^*}\chi\left(\frac{y}{a}\left(b+\frac{az}{y}\right)^{p^l+1}\right)\\
&=&\sum_{y,z\in \gf(p)^*}\sum_{b\in \gf(q)}\chi\left(\frac{y}{a}\left(b+\frac{az}{y}\right)^{p^l+1}\right)-\sum_{y,z\in \gf(p)^*}\chi\left(\frac{az^2}{y}\right)\\
&=&-(p-1)^2+\sum_{y,z\in \gf(p)^*}\sum_{b'\in \gf(q)}\chi\left(\frac{y}{a}b'^{p^l+1}\right)\\
&=&\sum_{y,z\in \gf(p)^*}\sum_{b'\in \gf(q)^*}\chi\left(\frac{y}{a}b'^{p^l+1}\right)\\
&=&(p^l+1)\sum_{y,z\in \gf(p)^*}\sum_{b''\in \gf(p^l)^*}\chi\left(\frac{y}{a}b''\right)=-(p^l+1)(p-1)^2,
\end{eqnarray*}
where
$$\chi\left(-\frac{az^2}{y}\right)=\zeta_{p}^{\tr_{p^l/p}\left(-\frac{az^2}{y}\right)}=\zeta_{p}^{-\frac{z^2}{y}\tr_{p^l/p}(a)}=1\mbox{ as }\tr_{p^l/p}(a)=0,$$
and we made the substitution $b+\frac{az}{y}\mapsto b'$ in the fifth equality.
By Equation (\ref{number-1-2}), we then have
\begin{eqnarray}\label{number-1-3}
\sharp \left\{b\in \gf(q)^*:\tr_{p^l/p}\left(\frac{b^{p^l+1}}{a}\right)=0\mbox{ and }\tr_{q/p}(b)=0\right\}=p^{2l-2}-p^l+p^{l-1}-1.
\end{eqnarray}
Note that Equation (\ref{number-1-3}) holds for any $a\in \gf(p^l)^*$ such that $\tr_{p^l/p}(a)=0$. Combining Equations (\ref{number-1}) and (\ref{number-1-3}) yields
$$N_1=(p^{l-1}-1)(p^{2l-2}-p^l+p^{l-1}-1)+(p^{l-1}-1)=(p^{l-1}-1)(p^{2l-2}-p^l+p^{l-1}).$$
The proof is now completed.
\end{proof}

\begin{lemma}\label{lem-N2-N6}
Let $m=2l$ with $l\geq 2$. Denote by
\begin{enumerate}
\item $N_2=\sharp \left\{(a,b)\in \gf(p^l)^*\times \gf(q):\tr_{p^l/p}\left(\frac{b^{p^l+1}}{a}\right)\neq 0,\tr_{p^l/p}(a)=0\mbox{ and }\tr_{q/p}(b)=0\right\}$,
\item $N_3=\sharp \left\{(a,b)\in \gf(p^l)^*\times \gf(q):\substack{\tr_{p^l/p}\left(\frac{b^{p^l+1}}{a}\right)=0\mbox{ and exactly one of }\\\tr_{p^l/p}(a)\mbox{ and }\tr_{q/p}(b)\mbox{ equals }0}\right\}$,
\item $N_4=\sharp \left\{(a,b)\in \gf(p^l)^*\times \gf(q):\substack{\tr_{p^l/p}\left(\frac{b^{p^l+1}}{a}\right)\neq 0\mbox{ and exactly one of }\\\tr_{p^l/p}(a)\mbox{ and }\tr_{q/p}(b)\mbox{ equals }0}\right\}$,
\item $N_5=\sharp \left\{(a,b)\in \gf(p^l)^*\times \gf(q):\tr_{p^l/p}\left(\frac{b^{p^l+1}}{a}\right)= 0,\tr_{p^l/p}(a)\neq 0\mbox{ and }\tr_{q/p}(b)\neq 0\right\}$,
\item $N_6=\sharp \left\{(a,b)\in \gf(p^l)^*\times \gf(q):\tr_{p^l/p}\left(\frac{b^{p^l+1}}{a}\right)\neq 0,\tr_{p^l/p}(a)\neq 0\mbox{ and }\tr_{q/p}(b)\neq 0\right\}$.
\end{enumerate}
Then
\begin{eqnarray*}\left\{
\begin{array}{l}
N_2=(p-1)(p^{l-1}-1)(p^{2l-2}+p^{l-1}),\\
N_3=p^{2l-2}(p-1)(2p^{l-1}-1),\\
N_4=p^{2l-2}(p-1)^2(2p^{l-1}-1),\\
N_5=p^{2l-2}(p-1)^2(p^{l-1}-1),\\
N_6=p^{2l-2}(p-1)^2(p^{l}-p^{l-1}+1).
\end{array} \right.
\end{eqnarray*}
\end{lemma}
\begin{proof}
Let $N_1$ be defined as that in Lemma \ref{lem-N1}. By definition, we have
$$N_1+N_2=\sharp \left\{(a,b)\in \gf(p^l)^*\times \gf(q):\tr_{p^l/p}(a)=0\mbox{ and }\tr_{q/p}(b)=0\right\}=p^{2l-1}(p^{l-1}-1)$$
$$\Rightarrow N_2=(p-1)(p^{l-1}-1)(p^{2l-2}+p^{l-1}).$$
Denote by $N_3=N_3^{(1)}+N_3^{(2)}$, where
$$N_3^{(1)}=\sharp \left\{(a,b)\in \gf(p^l)^*\times \gf(q):\tr_{p^l/p}\left(\frac{b^{p^l+1}}{a}\right)=0,\tr_{p^l/p}(a)\neq0\mbox{ and }\tr_{q/p}(b)=0\right\},$$ and
$$N_3^{(2)}=\sharp \left\{(a,b)\in \gf(p^l)^*\times \gf(q):\tr_{p^l/p}\left(\frac{b^{p^l+1}}{a}\right)=0,\tr_{p^l/p}(a)=0\mbox{ and }\tr_{q/p}(b)\neq0\right\}.$$
Similarly to the proof of Lemma \ref{lem-N1}, we can derive that
$$N_3^{(1)}=p^{2l-2}(p^{l}-p^{l-1}).$$
Observe that
\begin{eqnarray}\label{eqn-N3(2)}
\nonumber \lefteqn{ N_3^{(2)}+N_1 } \\
\nonumber &=&\sharp \left\{(a,b)\in \gf(p^l)^*\times \gf(q):\tr_{p^l/p}\left(\frac{b^{p^l+1}}{a}\right)=0\mbox{ and }\tr_{p^l/p}(a)=0\right\}\\
\nonumber &=&\sharp \left\{(a,b)\in \gf(p^l)^*\times \gf(q)^*:\tr_{p^l/p}\left(\frac{b^{p^l+1}}{a}\right)=0\mbox{ and }\tr_{p^l/p}(a)=0\right\}\\
& &+(p^{l-1}-1).
\end{eqnarray}
For fixed $a\in \gf(p^l)^*$ such that $\tr_{p^l/p}(a)=0$,
$$\tr_{p^l/p}\left(\frac{b^{p^l+1}}{a}\right)=0\Leftrightarrow b^{p^l+1} \in a\ker(\tr_{p^l/p})\setminus \{0\}.$$
 Since the norm function is a surjective homomorphism,
 $$\sharp \left\{(a,b)\in \gf(p^l)^*\times \gf(q)^*:\tr_{p^l/p}\left(\frac{b^{p^l+1}}{a}\right)=0\mbox{ and }\tr_{p^l/p}(a)=0\right\}=(p^l+1)(p^{l-1}-1)^2.$$
By Equation (\ref{eqn-N3(2)}), we have $N_3^{(2)}+N_1=(p^l+1)(p^{l-1}-1)^2+(p^{l-1}-1)$. Then
$$N_3^{(2)}=(p^l+1)(p^{l-1}-1)^2+(p^{l-1}-1)-N_1=(p^{l-1}-1)(p^{2l-1}-p^{2l-2})$$  by Lemma \ref{lem-N1}.
We then have $$N_3=N_3^{(1)}+N_3^{(2)}=p^{2l-2}(p-1)(2p^{l-1}-1).$$
It is easy to deduce that $N_3+N_4=(p^{l-1}-1)(p^{2l}-p^{2l-1})+p^{2l-1}(p^{l}-p^{l-1})$. We directly have
$$N_4=p^{2l-2}(p-1)^2(2p^{l-1}-1).$$
It is observed that
\begin{eqnarray*}
N_1+N_3+N_5&=&\sharp \left\{(a,b)\in \gf(p^l)^*\times \gf(q):\tr_{p^l/p}\left(\frac{b^{p^l+1}}{a}\right)= 0\right\}\\
&=&(p^l-1)+\sharp \left\{(a,b)\in \gf(p^l)^*\times \gf(q)^*:\tr_{p^l/p}\left(\frac{b^{p^l+1}}{a}\right)= 0\right\}\\
&=&(p^l-1)+(p^l-1)(p^l+1)(p^{l-1}-1),
\end{eqnarray*}
which implies
$$N_5=(p^l-1)+(p^l-1)(p^l+1)(p^{l-1}-1)-(N_1+N_3)=p^{2l-2}(p-1)^2(p^{l-1}-1).$$
Note that
\begin{eqnarray*}
N_5+N_6&=&\sharp \left\{(a,b)\in \gf(p^l)^*\times \gf(q):\tr_{p^l/p}(a)\neq 0\mbox{ and }\tr_{q/p}(b)\neq 0\right\}\\
&=&p^{3l-2}(p-1)^2.
\end{eqnarray*}
Then $N_6=p^{3l-2}(p-1)^2-N_5=p^{2l-2}(p-1)^2(p^{l}-p^{l-1}+1)$. The proof is completed.
\end{proof}

\begin{theorem}\label{th-code1}
Let $m=2l$ and $f(x)=x^{p^l+1}$ with $l\geq 2$. Then the $p$-ary subfield code $\C_{(f,q)}^{(p)}$  has parameters $[p^{2l}+1,3l,p^{l-1}(p^{l+1}-p^{l}-1)]$ and weight enumerator
\begin{eqnarray*}
1+(p-1)(p^{l-1}-1)(p^{2l-2}+p^{l-1})z^{p^{l-1}(p^{l+1}-p^{l}-1)} + \\
p^{2l-2}(p-1)^2(2p^{l-1}-1)  z^{p^{l-1}(p^{l+1}-p^{l}-1)+1} +\\
p^{2l-2}(p-1)^2(p^{l}-p^{l-1}+1)z^{p^{l-1}(p^{l+1}-p^{l}-1)+2}+ \\
(p^{2l-1}-1)z^{p^{2l-1}(p-1)} +  (p^{2l}-p^{2l-1}) z^{p^{2l-1}(p-1)+1} +  \\
(p^{l-1}-1)(p^{2l-2}-p^l+p^{l-1}) z^{(p-1)(p^{2l-1}+p^{l-1})}  + \\
p^{2l-2}(p-1)(2p^{l-1}-1) z^{(p-1)(p^{2l-1}+p^{l-1})+1}  + \\
p^{2l-2}(p-1)^2(p^{l-1}-1) z^{(p-1)(p^{2l-1}+p^{l-1})+2}.
\end{eqnarray*}
Its dual is nearly optimal with respect to the sphere-packing bound, and has parameters $[p^{2l}+1,p^{2l}+1-3l,3]$.
\end{theorem}

\begin{proof}
Let ${\bc_{(f,q)}}^{(p)}=\left(\left(\tr_{p^l/p}(ax^{p^l+1})+\tr_{q/p}(bx)\right)_{x\in \gf(q)^*},\tr_{p^l/p}(a),\tr_{q/p}(b)\right)$ be any codeword in $\C_{(f,q)}^{(p)}$.
By the orthogonality relation of additive characters,
\begin{eqnarray}\label{eqn-1}
\nonumber
\lefteqn{ \sharp \{x\in \gf(q)^*:\tr_{p^l/p}(ax^{p^l+1})+\tr_{q/p}(bx)\neq 0\} }\\
\nonumber&=&\sharp \{x\in \gf(q):\tr_{p^l/p}(ax^{p^l+1})+\tr_{q/p}(bx)\neq 0\}\\
\nonumber&=&q-\frac{1}{p}\sum_{x\in \gf(q)}\sum_{y\in \gf(p)}\zeta_{p}^{y\left(\tr_{p^l/p}(ax^{p^l+1})+\tr_{q/p}(bx)\right)}\\
\nonumber&=&q-\frac{1}{p}\sum_{x\in \gf(q)}\sum_{y\in \gf(p)}\zeta_{p}^{y\left(\tr_{p^l/p}(ax^{p^l+1}+bx+b^{p^l}x^{p^l})\right)}\\
&=&\frac{(p-1)q}{p}-\frac{1}{p}\sum_{y\in \gf(p)^*}\sum_{x\in \gf(q)}\chi(yax^{p^l+1}+ybx+yb^{p^l}x^{p^l}),
\end{eqnarray}
where $\chi$ denotes the canonical additive character of $\gf(p^l)$.

If $a\neq 0$, then
$$yax^{p^l+1}+ybx+yb^{p^l}x^{p^l}=ya\left(x^{p^l+1}+\frac{b}{a}x+\frac{b^{p^l}x^{p^l}}{a}\right)=ya\left(\left(x+\frac{b^{p^l}}{a}\right)^{p^l+1}-\frac{b^{p^l+1}}{a^2}\right)$$
for $y\in \gf(p)^*$. By Equation (\ref{eqn-1}),
\begin{eqnarray*}
\lefteqn{ \sharp \{x\in \gf(q)^*:\tr_{p^l/p}(ax^{p^l+1})+\tr_{q/p}(bx)\neq 0\} } \\
&=&\frac{(p-1)q}{p}-\frac{1}{p}\sum_{y\in \gf(p)^*}\sum_{x\in \gf(q)}\chi\left(ya\left(\left(x+\frac{b^{p^l}}{a}\right)^{p^l+1}-\frac{b^{p^l+1}}{a^2}\right)\right)\\
&=&\frac{(p-1)q}{p}-\frac{1}{p}\sum_{y\in \gf(p)^*}\chi\left(-\frac{b^{p^l+1}}{a}y\right)\sum_{x\in \gf(q)}\chi\left(ya\left(x+\frac{b^{p^l}}{a}\right)^{p^l+1}\right)\\
&=&\frac{(p-1)q}{p}-\frac{1}{p}\sum_{y\in \gf(p)^*}\chi\left(-\frac{b^{p^l+1}}{a}y\right)\sum_{x'\in \gf(q)}\chi\left(yax'^{p^l+1}\right),
\end{eqnarray*}
where we made the  substitution $x+\frac{b^{p^l}}{a}\mapsto x'$ in the last equality. Note that $\Norm(x')=x'^{p^l+1}$ is the norm function from $\gf(q)^*$ onto $\gf(p^l)^*$. Therefore,
$$\sum_{x'\in \gf(q)}\chi\left(yax'^{p^l+1}\right)=1+\sum_{x'\in \gf(q)^*}\chi\left(yax'^{p^l+1}\right)=1+(p^l+1)\sum_{z\in \gf(p^l)^*}\chi(yaz)=-p^l.$$
Then we have
\begin{eqnarray*}
\lefteqn{ \sharp \{x\in \gf(q)^*:\tr_{p^l/p}(ax^{p^l+1})+\tr_{q/p}(bx)\neq 0\}  } \\
&=&\frac{(p-1)q}{p}+p^{l-1}\sum_{y\in \gf(p)^*}\chi\left(-\frac{b^{p^l+1}}{a}y\right)\\
&=&(p-1)p^{2l-1}+p^{l-1}\sum_{y\in \gf(p)^*}\zeta_{p}^{-y\tr_{p^l/p}\left(\frac{b^{p^l+1}}{a}\right)}\\
&=&\left\{
\begin{array}{lll}
(p-1)(p^{2l-1}+p^{l-1})    &   \mbox{ if }\tr_{p^l/p}\left(\frac{b^{p^l+1}}{a}\right)=0,\\
p^{l-1}(p^{l+1}-p^{l}-1)   &   \mbox{ if }\tr_{p^l/p}\left(\frac{b^{p^l+1}}{a}\right)\neq 0,\\
\end{array} \right.
\end{eqnarray*}
by the orthogonality relation of additive characters. By definition, we have
\begin{eqnarray}\label{eqn-weights}
\wt(\bc_{(f,q)}^{(p)})&=&
\left\{
\begin{array}{lll}
(p-1)(p^{2l-1}+p^{l-1})    &   \substack{\mbox{ if }\tr_{p^l/p}\left(\frac{b^{p^l+1}}{a}\right)=0\mbox{ and }\tr_{p^l/p}(a)=0,\\\tr_{q/p}(b)=0,}\\
p^{l-1}(p^{l+1}-p^{l}-1)   &   \substack{\mbox{ if }\tr_{p^l/p}\left(\frac{b^{p^l+1}}{a}\right)\neq 0\mbox{ and }\tr_{p^l/p}(a)=0,\\\tr_{q/p}(b)=0,}\\
(p-1)(p^{2l-1}+p^{l-1})+1    &   \substack{\mbox{ if }\tr_{p^l/p}\left(\frac{b^{p^l+1}}{a}\right)=0\mbox{ and exactly one of }\\\tr_{p^l/p}(a)\mbox{ and }\tr_{q/p}(b)\mbox{ equals }0,}\\
p^{l-1}(p^{l+1}-p^{l}-1)+1   &   \substack{\mbox{ if }\tr_{p^l/p}\left(\frac{b^{p^l+1}}{a}\right)\neq0\mbox{ and exactly one of }\\\tr_{p^l/p}(a)\mbox{ and }\tr_{q/p}(b)\mbox{ equals }0,}\\
(p-1)(p^{2l-1}+p^{l-1})+2    &   \substack{\mbox{ if }\tr_{p^l/p}\left(\frac{b^{p^l+1}}{a}\right)=0\mbox{ and }\tr_{p^l/p}(a)\neq 0,\\\tr_{q/p}(b)\neq 0,}\\
p^{l-1}(p^{l+1}-p^{l}-1)+2   &   \substack{\mbox{ if }\tr_{p^l/p}\left(\frac{b^{p^l+1}}{a}\right)\neq 0\mbox{ and }\tr_{p^l/p}(a)\neq 0,\\ \tr_{q/p}(b)\neq 0.}
\end{array} \right.
\end{eqnarray}

If $a=0$, then the codeword ${\bc_{(f,q)}}^{(p)}=\left(\left(\tr_{q/p}(bx)\right)_{x\in \gf(q)^*},0,\tr_{q/p}(b)\right)$. It is easy to deduce that
\begin{eqnarray}\label{eqn-weights-2}
\wt(\bc_{(f,q)}^{(p)})&=&
\left\{
\begin{array}{lll}
0    &   \mbox{ if }a=b=0,\\
p^{2l-1}(p-1)   &  \mbox{ if }a=0,b\neq 0,\tr_{q/p}(b)=0,\\
p^{2l-1}(p-1)+1   &  \mbox{ if }a=0,\tr_{q/p}(b)\neq 0.
\end{array} \right.
\end{eqnarray}

Due to Equations (\ref{eqn-weights}) and (\ref{eqn-weights-2}), we deduce that the minimal distance $d^{(p)}$ of  $\C_{(f,q)}^{(p)}$ satisfies $d^{(p)}\geq p^{l-1}(p^{l+1}-p^{l}-1)$. By Lemma \ref{lem-N2-N6} and Equation (\ref{eqn-weights}),
$$A_{p^{l-1}(p^{l+1}-p^{l}-1) }=N_2=(p-1)(p^{l-1}-1)(p^{2l-2}+p^{l-1})>0\mbox{ for }l\geq2.$$
Therefore, the minimal distance $d^{(p)}=p^{l-1}(p^{l+1}-p^{l}-1)$. The dimension of $\C_{(f, q)}^{(p)}$ is $3l$ as $\wt(\bc_{(f,q)}^{(p)})=0$ if and only if $a=b=0$ for $l\geq 2$ by Equations (\ref{eqn-weights}) and (\ref{eqn-weights-2}). The parameters of $\C_{(f, q)}^{(p)}$ follow. Note that the frequency of each weight in Equation (\ref{eqn-weights-2}) is easy to derive. Then the weight distribution of  $\C_{(f,q)}^{(p)}$ follows from Lemmas \ref{lem-N1} and \ref{lem-N2-N6}.

By Theorem \ref{th-dualdistance}, the minimal distance $d^{(p)\perp}$ of $\C_{(f, q)}^{(p)\perp}$ satisfies $d^{(p)\perp}\geq 3$ as the dual of $\C_{(f, q)}$ has minimal distance 3.  From the weight distribution of $\C_{(f,q)}^{(p)}$ and the first four Pless power moments in \cite[Page 131]{HP}, we can prove that $A_3^{(p)\perp}>0$, where $A_3^{(p)\perp}$ denotes the number of the codewords with weight 3 in $\C_{(f, q)}^{(p)\perp}$. Then the parameters of $\C_{(f, q)}^{(p)\perp}$ follow. By the sphere-packing bound, one can deduce that $d^{(p)\perp}\leq 4$. Hence the dual of $\C_{(f, q)}^{(p)}$ is nearly optimal with respect to the sphere-packing bound.
\end{proof}

Theorem \ref{th-code1} shows that the code $\C_{(f, q)}^{(p)}$ is projective as its dual has minimal distance 3. The following example shows that $\C_{(f, q)}^{(p)}$ has very good parameters.

\begin{example} Let $m=2l$ and $f(x)=x^{p^l+1}$ with $l\geq 2$.
\begin{enumerate}
\item Let $p=2$ and $l=2$. Then the set $\C_{(f, q)}^{(p)}$ in Theorem \ref{th-code1} is a $[17,6,6]$ binary code whose dual is a $[17,11,3]$ binary code, while the corresponding best known parameters are $[17,6,7]$ and $[17,11,4]$ according to the Code Tables at http://www.codetables.de/.
\item Let $p=2$ and $l=3$. Then the set $\C_{(f, q)}^{(p)}$ in Theorem \ref{th-code1} is a $[65,9,28]$ binary code whose dual is a $[65,56,3]$ binary code, while the corresponding best known parameters are $[65,9,28]$ and $[65,56,4]$ according to the Code Tables at http://www.codetables.de/.
\item Let $p=3$ and $l=2$. Then the set $\C_{(f, q)}^{(p)}$ in Theorem \ref{th-code1} is a $[82,6,51]$ ternary code, while the corresponding best known parameters are $[82,6,52]$ according to the Code Tables at http://www.codetables.de/. Its dual is a $[82,76,3]$ ternary code which has the best known parameters according to the Code Tables at http://www.codetables.de/.
\end{enumerate}
\end{example}

\subsection{The subfield code $\C_{(f,q)}^{(p)}$ for $f(x)=x^2$ and odd $p$}
Let $f(x)=x^2$ and $p$ be odd. Then $\gcd(q-1,2-1)=1$ and $\C_{(f, q)}$ is a $[q+1, 2, q]$ MDS code  by Lemma \ref{lem-monomial}. It is known that $f(x)=x^2$ is a planar function over $\gf(q)$.  By Theorem \ref{thm-tracerepresent},
the $p$-ary subfield code of $\C_{(f, q)}$ is given by
$$\C_{(f,q)}^{(p)}=\left\{{\bc_{(f,q)}}^{(p)}=\left(\left(\tr_{q/p}(ax^2+bx)\right)_{x\in \gf(q)^*},\tr_{q/p}(a),\tr_{q/p}(b)\right):\substack{a\in \gf(q)\\b\in \gf(q)}\right\}.$$

\begin{lemma}\label{lem-N1-N4} Let $m$ be odd. The followings hold.
\begin{enumerate}
\item \begin{eqnarray*}&N_1:=\sharp \left\{(a,b)\in \gf(q)\times \gf(q):a\neq 0, \tr_{q/p}(\frac{b^2}{4a})=0, \tr_{q/p}(a)=0, \tr_{q/p}(b)= 0\right\}\\
    &=(p^{m-1}-1)p^{m-2}.\end{eqnarray*}
\item \begin{eqnarray*}&N_2:=\sharp \left\{(a,b)\in \gf(q)\times \gf(q):a\neq 0, \tr_{q/p}(\frac{b^2}{4a})=0, \tr_{q/p}(a)=0, \tr_{q/p}(b)\neq  0\right\}\\
    &=(p^{m-1}-1)(p^{m-1}-p^{m-2}).\end{eqnarray*}
\item \begin{eqnarray*}&N_3:=\sharp \left\{(a,b)\in \gf(q)\times \gf(q):a\neq 0, \tr_{q/p}(\frac{b^2}{4a})=0, \tr_{q/p}(a)\neq 0, \tr_{q/p}(b)= 0\right\}\\
    &=p^{m-2}(p-1)(p^{m-1}+p-1).\end{eqnarray*}
\item \begin{eqnarray*}&N_4:=\sharp \left\{(a,b)\in \gf(q)\times \gf(q):a\neq 0, \tr_{q/p}(\frac{b^2}{4a})=0, \tr_{q/p}(a)\neq 0, \tr_{q/p}(b)\neq  0\right\}\\
    &=p^{m-2}(p-1)^2(p^{m-1}-1).\end{eqnarray*}
\end{enumerate}
\end{lemma}

\begin{proof}
Firstly, we compute
 $$N_1=\sharp \left\{(a,b)\in \gf(q)\times \gf(q):a\neq 0, \tr_{q/p}(\frac{b^2}{4a})=0, \tr_{q/p}(a)=0, \tr_{q/p}(b)= 0\right\}.$$
For fixed nonzero $a$ satisfying $\tr_{q/p}(a)=0$, we have
\begin{eqnarray*}
\lefteqn{ \sharp \left\{b\in \gf(q):\tr_{q/p}(\frac{b^2}{4a})=0\mbox{ and }\tr_{q/p}(b)= 0\right\} } \\
&=&\frac{1}{p^2}\sum_{b\in \gf(q)}\sum_{y\in \gf(p)}\chi(\frac{yb^2}{4a})\sum_{z\in \gf(p)}\chi(zb)\\
&=&p^{m-2}+\frac{1}{p^2}\sum_{z\in \gf(p)^*}\sum_{b\in \gf(q)}\chi(zb)+\frac{1}{p^2}\sum_{b\in \gf(q)}\sum_{y\in \gf(p)^*}\chi(\frac{yb^2}{4a})\\
& &+\frac{1}{p^2}\sum_{b\in \gf(q)}\sum_{y\in \gf(p)^*}\chi(\frac{yb^2}{4a})\sum_{z\in \gf(p)^*}\chi(zb)\\
&=&p^{m-2}+0+\frac{1}{p^2}\sum_{y\in \gf(p)^*}\sum_{b\in \gf(q)}\chi(\frac{yb^2}{4a})+\frac{1}{p^2}\sum_{y\in \gf(p)^*}\sum_{z\in \gf(p)^*}\sum_{b\in \gf(q)}\chi(\frac{yb^2}{4a}+zb),
\end{eqnarray*}
where $\chi$ denotes the canonical additive character of $\gf(q)$. Let $\eta,\eta'$ be the quadratic multiplicative  character of $\gf(q)^*$ and $\gf(p)^*$, respectively.
By Lemma \ref{lem-charactersum} and the orthogonality relation of  multiplicative  characters,
\begin{eqnarray}\label{eqn-part}
\sum_{y\in \gf(p)^*}\sum_{b\in \gf(q)}\chi(\frac{yb^2}{4a})=G(\eta,\chi)\sum_{y\in \gf(p)^*}\eta(\frac{y}{a})=G(\eta,\chi)\eta(a^{-1})\sum_{y\in \gf(p)^*}\eta'(y)=0,
\end{eqnarray}
\begin{eqnarray*}
\sum_{y\in \gf(p)^*}\sum_{z\in \gf(p)^*}\sum_{b\in \gf(q)}\chi(\frac{yb^2}{4a}+zb)&=&G(\eta,\chi)\sum_{y\in \gf(p)^*}\eta(\frac{y}{a})\sum_{z\in \gf(p)^*}\chi(-\frac{a}{y}z^2)\\
&=&G(\eta,\chi)\eta(a^{-1})\sum_{y\in \gf(p)^*}\eta(y)\sum_{z\in \gf(p)^*}\zeta_{p}^{-\frac{z^2}{y}\tr_{q/p}(a)}\\
&=&(p-1)G(\eta,\chi)\eta(a^{-1})\sum_{y\in \gf(p)^*}\eta'(y)=0
\end{eqnarray*}
as $\eta(y)=\eta'(y)$ for $y\in \gf(p)^*$ and $\tr_{q/p}(a)=0$. Hence
$$\left\{b\in \gf(q):\tr_{q/p}(\frac{b^2}{4a})=0\mbox{ and }\tr_{q/p}(b)= 0\right\}=p^{m-2}$$ \mbox{ for any fixed nonzero $a$ satisfying }$\tr_{q/p}(a)=0$ and
$$N_1=(p^{m-1}-1)p^{m-2}.$$

Secondly, we compute
$$N_3=\sharp \left\{(a,b)\in \gf(q)\times \gf(q):a\neq 0, \tr_{q/p}(\frac{b^2}{4a})=0, \tr_{q/p}(a)\neq 0, \tr_{q/p}(b)= 0\right\}.$$
By definition,
\begin{eqnarray*}
N_1+N_3&=&\sharp \left\{(a,b)\in \gf(q)\times \gf(q):a\neq 0, \tr_{q/p}(\frac{b^2}{4a})=0, \tr_{q/p}(b)= 0\right\}\\
&=&\frac{1}{p^2}\sum_{a\in \gf(q)^*}\sum_{b\in \gf(q)}\sum_{y\in \gf(p)}\chi(\frac{b^2y}{4a})\sum_{z\in \gf(p)}\chi(bz)\\
&=&p^{m-2}(p^m-1)+\frac{1}{p^2}\sum_{a\in \gf(q)^*}\sum_{z\in \gf(p)^*}\sum_{b\in \gf(q)}\chi(zb)+\frac{1}{p^2}\sum_{a\in \gf(q)^*}\sum_{b\in \gf(q)}\sum_{y\in \gf(p)^*}\chi(\frac{yb^2}{4a})\\
& &+\frac{1}{p^2}\sum_{a\in \gf(q)^*}\sum_{b\in \gf(q)}\sum_{y\in \gf(p)^*}\chi(\frac{yb^2}{4a})\sum_{z\in \gf(p)^*}\chi(zb)\\
&=&p^{m-2}(p^m-1)+0+\frac{1}{p^2}\sum_{a\in \gf(q)^*}\sum_{y\in \gf(p)^*}\sum_{b\in \gf(q)}\chi(\frac{yb^2}{4a})\\
& &+\frac{1}{p^2}\sum_{y\in \gf(p)^*}\sum_{z\in \gf(p)^*}\sum_{a\in \gf(q)^*}\sum_{b\in \gf(q)}\chi(\frac{yb^2}{4a}+zb).
\end{eqnarray*}
By Equation (\ref{eqn-part}), we have  $$\sum_{a\in \gf(q)^*}\sum_{y\in \gf(p)^*}\sum_{b\in \gf(q)}\chi(\frac{yb^2}{4a})=0.$$
By the orthogonality relation of additive characters,
\begin{eqnarray*}
\lefteqn{ \sum_{y\in \gf(p)^*}\sum_{z\in \gf(p)^*}\sum_{a\in \gf(q)^*}\sum_{b\in \gf(q)}\chi(\frac{yb^2}{4a}+zb) } \\
&=&(p-1)^2(q-1)+\sum_{y\in \gf(p)^*}\sum_{z\in \gf(p)^*}\sum_{a\in \gf(q)^*}\sum_{b\in \gf(q)^*}\chi(\frac{yb^2}{4a}+zb)\\
&=&(p-1)^2(q-1)+\sum_{y\in \gf(p)^*}\sum_{z\in \gf(p)^*}\sum_{b\in \gf(q)^*}\chi(zb)\sum_{a\in \gf(q)^*}\chi(\frac{yb^2}{4a})\\
&=&(p-1)^2(q-1)-\sum_{y\in \gf(p)^*}\sum_{z\in \gf(p)^*}\sum_{b\in \gf(q)^*}\chi(zb)\\
&=&p^m(p-1)^2.
\end{eqnarray*}
Hence $N_1+N_3=p^{m-2}(p^m+p^2-2p)$ and then $N_3=p^{m-2}(p-1)(p^{m-1}+p-1)$.

The values of $N_2$ and $N_4$ can be easily determined by their connections with $N_1$ and $N_3$.
\end{proof}

\begin{lemma}\label{lem-N5-N12}
 Let $m$ be an odd positive integer and $q=p^m$.
\begin{enumerate}
\item Denote by $N_5$ the number of the solutions $(a,b)\in \gf(q)\times \gf(q)$ of
$$\left\{\begin{array}{l}
a\neq 0,\\
 \tr_{q/p}(\frac{b^2}{4a})\neq 0,\\
 \eta(a)\eta'\left(-\tr_{q/p}(\frac{b^2}{4a})\right)=1,\\
\tr_{q/p}(a)= 0,\\
\tr_{q/p}(b)= 0,\\
\end{array} \right.$$ then $N_5=\frac{(p-1)(p^{m-1}-1)(p^{m-2}+(-1)^{\frac{(p-1)(m+1)}{4}}p^{\frac{m-1}{2}})}{2}$.
\item Denote by $N_6$ the number of the solutions $(a,b)\in \gf(q)\times \gf(q)$ of
$$\left\{\begin{array}{l}
a\neq 0,\\
 \tr_{q/p}(\frac{b^2}{4a})\neq 0,\\
 \eta(a)\eta'\left(-\tr_{q/p}(\frac{b^2}{4a})\right)=-1,\\
\tr_{q/p}(a)= 0,\\
\tr_{q/p}(b)= 0,\\
\end{array} \right.$$ then $N_6=\frac{(p-1)(p^{m-1}-1)(p^{m-2}+(-1)^{\frac{(p-1)(m+1)+4}{4}}p^{\frac{m-1}{2}})}{2}$.
\item Denote by $N_7$ the number of the solutions $(a,b)\in \gf(q)\times \gf(q)$ of
$$\left\{\begin{array}{l}
a\neq 0,\\
 \tr_{q/p}(\frac{b^2}{4a})\neq 0,\\
 \eta(a)\eta'\left(-\tr_{q/p}(\frac{b^2}{4a})\right)=1,\\
\tr_{q/p}(a)= 0,\\
\tr_{q/p}(b)\neq  0,\\
\end{array} \right.$$ then $N_7=\frac{p^{m-2}(p-1)^{2}(p^{m-1}-1)}{2}$.
\item Denote by $N_8$ the number of the solutions $(a,b)\in \gf(q)\times \gf(q)$ of
$$\left\{\begin{array}{l}
a\neq 0,\\
 \tr_{q/p}(\frac{b^2}{4a})\neq 0,\\
 \eta(a)\eta'\left(-\tr_{q/p}(\frac{b^2}{4a})\right)=1,\\
\tr_{q/p}(a)\neq 0,\\
\tr_{q/p}(b)=  0,\\
\end{array} \right.$$ then $N_8=\frac{p^{m-2}(p-1)^{2}(p^{m-1}-1)}{2}$.
\item Denote by $N_9$ the number of the solutions $(a,b)\in \gf(q)\times \gf(q)$ of
$$\left\{\begin{array}{l}
a\neq 0,\\
 \tr_{q/p}(\frac{b^2}{4a})\neq 0,\\
 \eta(a)\eta'\left(-\tr_{q/p}(\frac{b^2}{4a})\right)=1,\\
\tr_{q/p}(a)\neq 0,\\
\tr_{q/p}(b)\neq 0,\\
\end{array} \right.$$ then $N_9=\frac{p^{m-2}(p-1)^{2}(p^m-p^{m-1}+1)+(-1)^{\frac{(p-1)(m+1)}{4}}p^{\frac{3(m-1)}{2}}(p-1)^{2}}{2}$.
\item Denote by $N_{10}$ the number of the solutions $(a,b)\in \gf(q)\times \gf(q)$ of
$$\left\{\begin{array}{l}
a\neq 0,\\
 \tr_{q/p}(\frac{b^2}{4a})\neq 0,\\
 \eta(a)\eta'\left(-\tr_{q/p}(\frac{b^2}{4a})\right)=-1,\\
\tr_{q/p}(a)= 0,\\
\tr_{q/p}(b)\neq  0,\\
\end{array} \right.$$ then $N_{10}=\frac{p^{m-2}(p-1)^{2}(p^{m-1}-1)}{2}$.
\item Denote by $N_{11}$ the number of the solutions $(a,b)\in \gf(q)\times \gf(q)$ of
$$\left\{\begin{array}{l}
a\neq 0,\\
 \tr_{q/p}(\frac{b^2}{4a})\neq 0,\\
 \eta(a)\eta'\left(-\tr_{q/p}(\frac{b^2}{4a})\right)=-1,\\
\tr_{q/p}(a)\neq 0,\\
\tr_{q/p}(b)=  0,\\
\end{array} \right.$$ then $N_{11}=\frac{p^{m-2}(p-1)^{2}(p^{m-1}-1)}{2}$.
\item Denote by $N_{12}$ the number of the solutions $(a,b)\in \gf(q)\times \gf(q)$ of
$$\left\{\begin{array}{l}
a\neq 0,\\
 \tr_{q/p}(\frac{b^2}{4a})\neq 0,\\
 \eta(a)\eta'\left(-\tr_{q/p}(\frac{b^2}{4a})\right)=-1,\\
\tr_{q/p}(a)\neq 0,\\
\tr_{q/p}(b)\neq  0,\\
\end{array} \right.$$ then $N_{12}=\frac{p^{m-2}(p-1)^{2}(p^m-p^{m-1}+1)+(-1)^{\frac{(p-1)(m+1)+4}{4}}p^{\frac{3(m-1)}{2}}(p-1)^{2}}{2}$.
\end{enumerate}
\end{lemma}
\begin{proof}
Suppose that $\gf(q)^*=\langle\alpha\rangle$ and $\gf(p)^*=\langle\beta\rangle$. Let $C_0,C_0'$ be the cyclic groups generated by $\alpha^2$ and $\beta^2$, respectively. Denote by $\chi$ and $\chi'$  the canonical additive characters of $\gf(q)$ and $\gf(p)$, respectively. Let $\eta,\eta'$ be the quadratic multiplicative characters of $\gf(q)$ and $\gf(p)$, respectively.

Firstly, we determine $N_5$.
\begin{enumerate}
\item \label{item111} If $\eta(a)=\eta'\left(-\tr_{q/p}(\frac{b^2}{4a})\right)=1$, then $a\in C_0$ and $\tr_{q/p}(\frac{b^2}{4a})\in -C_0'$. Now fix nonzero $a$ such that $\tr_{q/p}(a)= 0$ and $\eta(a)=1$. Let $\tr_{q/p}(\frac{b^2}{4a})+\beta^{2t}=0$ for some $0\leq t\leq \frac{p-3}{2}$. Then
\begin{eqnarray*}
\lefteqn{ \sharp \{b\in \gf(q):\tr_{q/p}(\frac{b^2}{4a})+\beta^{2t}=0 \mbox{ and }\tr_{q/p}(b)= 0\} } \\
&=&\frac{1}{p^2}\sum_{b\in \gf(q)}\sum_{y\in \gf(p)}\zeta_{p}^{y\tr_{q/p}(\frac{b^2}{4a})+y\beta^{2t}}\sum_{z\in \gf(p)}\chi(zb)\\
&=&p^{m-2}+0+\frac{1}{p^2}\sum_{y\in \gf(p)^*}\chi'(y\beta^{2t})\sum_{b\in \gf(q)}\chi(\frac{yb^2}{4a})\\
& &+\frac{1}{p^2}\sum_{y\in \gf(p)^*}\chi'(y\beta^{2t})\sum_{z\in \gf(p)^*}\sum_{b\in \gf(q)}\chi(\frac{yb^2}{4a}+zb).
\end{eqnarray*}
By Lemmas \ref{lem-charactersum} and \ref{quadGuasssum},
\begin{eqnarray*}
\lefteqn{ \sum_{y\in \gf(p)^*}\chi'(y\beta^{2t})\sum_{b\in \gf(q)}\chi(\frac{yb^2}{4a}) } \\
&=& G(\eta,\chi)\sum_{y\in \gf(p)^*}\chi'(y\beta^{2t})\eta'(y\beta^{2t})  \\
&=&G(\eta,\chi)G(\eta',\chi') \\
&=& (-1)^{\frac{(p-1)(m+1)}{4}}p^{\frac{m+1}{2}},
\end{eqnarray*}
\begin{eqnarray*}
\lefteqn{ \sum_{y\in \gf(p)^*}\chi'(y\beta^{2t})\sum_{z\in \gf(p)^*}\sum_{b\in \gf(q)}\chi(\frac{yb^2}{4a}+zb) } \\
&=&G(\eta,\chi)\sum_{y\in \gf(p)^*}\chi'(y\beta^{2t})\eta'(y\beta^{2t})\sum_{z\in \gf(p)^*}\zeta_{p}^{-\frac{z^2}{y}\tr_{q/p}(a)}\\
&=&(p-1)G(\eta,\chi)G(\eta',\chi')=(-1)^{\frac{(p-1)(m+1)}{4}}p^{\frac{m+1}{2}}(p-1),
\end{eqnarray*}
due to $\tr_{q/p}(a)= 0$ and $\eta(a)=1$. Then
$$\sharp \{b\in \gf(q):\tr_{q/p}(\frac{b^2}{4a})+\beta^{2t}=0 \mbox{ and }\tr_{q/p}(b)= 0\}=p^{m-2}+(-1)^{\frac{(p-1)(m+1)}{4}}p^{\frac{m-1}{2}}$$
for all $0\leq t\leq \frac{p-3}{2}$ and any fixed nonzero $a$ such that $\tr_{q/p}(a)= 0$ and $\eta(a)=1$.
\item If $\eta(a)=\eta'\left(-\tr_{q/p}(\frac{b^2}{4a})\right)=-1$, then $a\in \alpha C_0$ and $\tr_{q/p}(\frac{b^2}{4a})\in -\beta C_0'$. Now fix nonzero $a$ such that $\tr_{q/p}(a)= 0$ and $\eta(a)=-1$. Let $\tr_{q/p}(\frac{b^2}{4a})+\beta^{2t+1}=0$ for some $0\leq t\leq \frac{p-3}{2}$. Similarly to \ref{item111}), we can obtain that
    $$\sharp \{b\in \gf(q):\tr_{q/p}(\frac{b^2}{4a})+\beta^{2t}=0 \mbox{ and }\tr_{q/p}(b)= 0\}=p^{m-2}+(-1)^{\frac{(p-1)(m+1)}{4}}p^{\frac{m-1}{2}}$$
for all $0\leq t\leq \frac{p-3}{2}$ and any fixed nonzero $a$ such that $\tr_{q/p}(a)= 0$ and $\eta(a)=-1$.
\end{enumerate}
By \cite[Lemma 14]{HD19},
\begin{eqnarray*}
& \sharp \{a\in \gf(q^*):\eta(a)=1\mbox{ and }\tr_{q/p}(a)=0\}=\sharp \{a\in \gf(q^*):\eta(a)=-1\mbox{ and }\tr_{q/p}(a)=0\}\\
&=\frac{p^{m-1}-1}{2}.
\end{eqnarray*}
Therefore,
$$N_5=\frac{(p-1)(p^{m-1}-1)(p^{m-2}+(-1)^{\frac{(p-1)(m+1)}{4}}p^{\frac{m-1}{2}})}{2}.$$
Similarly, it can be proved that
$$N_6=\frac{(p-1)(p^{m-1}-1)(p^{m-2}+(-1)^{\frac{(p-1)(m+1)+4}{4}}p^{\frac{m-1}{2}})}{2}.$$

Secondly, we compute $N_7$. Note that $N_5+N_7$ equals the number of the solutions $(a,b)\in \gf(q)\times \gf(q)$ of
$$\left\{\begin{array}{l}
a\neq 0,\\
 \tr_{q/p}(\frac{b^2}{4a})\neq 0,\\
 \eta(a)\eta'\left(-\tr_{q/p}(\frac{b^2}{4a})\right)=1,\\
\tr_{q/p}(a)= 0.
\end{array} \right.$$
\begin{enumerate}
\item If $\eta(a)=\eta'\left(-\tr_{q/p}(\frac{b^2}{4a})\right)=1$, then $\tr_{q/p}(\frac{b^2}{4a})\in -C_0'$. Now fix nonzero $a$ such that $\tr_{q/p}(a)= 0$ and $\eta(a)=1$. Let $\tr_{q/p}(\frac{b^2}{4a})+\beta^{2t}=0$ for some $0\leq t\leq \frac{p-3}{2}$. By Lemmas \ref{lem-charactersum} and \ref{quadGuasssum},
    \begin{eqnarray*}
\lefteqn{ \sharp \{b\in \gf(q):\tr_{q/p}(\frac{b^2}{4a})+\beta^{2t}=0\} } \\
&=&\frac{1}{p}\sum_{b\in \gf(q)}\sum_{y\in \gf(p)}\zeta_{p}^{y\tr_{q/p}(\frac{b^2}{4a})+y\beta^{2t}}\\
&=&p^{m-1}+\frac{1}{p}\sum_{y\in \gf(p)^*}\chi'(y\beta^{2t})\sum_{b\in \gf(q)}\chi(\frac{yb^2}{4a})\\
&=&p^{m-1}+\frac{1}{p}G(\eta,\chi)G(\eta',\chi')\\
&=&p^{m-1}+(-1)^{\frac{(p-1)(m+1)}{4}}p^{\frac{m-1}{2}}
\end{eqnarray*}
for all $0\leq t\leq \frac{p-3}{2}$ and any fixed nonzero $a$ such that $\tr_{q/p}(a)= 0$ and $\eta(a)=1$.
\item If $\eta(a)=\eta'\left(-\tr_{q/p}(\frac{b^2}{4a})\right)=-1$, then $\tr_{q/p}(\frac{b^2}{4a})\in -\beta C_0'$. Now fix nonzero $a$ such that $\tr_{q/p}(a)= 0$ and $\eta(a)=-1$. Let $\tr_{q/p}(\frac{b^2}{4a})+\beta^{2t+1}=0$ for some $0\leq t\leq \frac{p-3}{2}$. By Lemmas \ref{lem-charactersum} and \ref{quadGuasssum},
    \begin{eqnarray*}
\lefteqn{ \sharp \{b\in \gf(q):\tr_{q/p}(\frac{b^2}{4a})+\beta^{2t+1}=0\} } \\
&=&\frac{1}{p}\sum_{b\in \gf(q)}\sum_{y\in \gf(p)}\zeta_{p}^{y\tr_{q/p}(\frac{b^2}{4a})+y\beta^{2t+1}}\\
&=&p^{m-1}+\frac{1}{p}\sum_{y\in \gf(p)^*}\chi'(y\beta^{2t+1})\sum_{b\in \gf(q)}\chi(\frac{yb^2}{4a})\\
&=&p^{m-1}+\frac{1}{p}G(\eta,\chi)G(\eta',\chi')\\
&=&p^{m-1}+(-1)^{\frac{(p-1)(m+1)}{4}}p^{\frac{m-1}{2}}
\end{eqnarray*}
for all $0\leq t\leq \frac{p-3}{2}$ and any fixed nonzero $a$ such that $\tr_{q/p}(a)= 0$ and $\eta(a)=-1$.
\end{enumerate}
By \cite[Lemma 14]{HD19} and the preceding discussions, we then deduce that
$$N_5+N_7=\frac{(p-1)(p^{m-1}-1)(p^{m-1}+(-1)^{\frac{(p-1)(m+1)}{4}}p^{\frac{m-1}{2}})}{2}$$ and
$$N_7=\frac{p^{m-2}(p-1)^{2}(p^{m-1}-1)}{2}.$$

Thirdly, we compute $N_8$.  Note that $N_5+N_8$ equals the number of the solutions $(a,b)\in \gf(q)\times \gf(q)$ of
$$\left\{\begin{array}{l}
a\neq 0,\\
 \tr_{q/p}(\frac{b^2}{4a})\neq 0,\\
 \eta(a)\eta'\left(-\tr_{q/p}(\frac{b^2}{4a})\right)=1,\\
\tr_{q/p}(b)= 0.
\end{array} \right.$$
\begin{enumerate}
\item \label{item11} If $\eta(a)=\eta'\left(-\tr_{q/p}(\frac{b^2}{4a})\right)=1$, then $a\in C_0$ and $\tr_{q/p}(\frac{b^2}{4a})\in -C_0'$. Let $\tr_{q/p}(\frac{b^2}{4a})+\beta^{2t}=0$ for some $0\leq t\leq \frac{p-3}{2}$. Then
\begin{eqnarray*}
\lefteqn{ \sharp \{(a,b)\in C_0\times  \gf(q):\tr_{q/p}(\frac{b^2}{4a})+\beta^{2t}=0 \mbox{ and }\tr_{q/p}(b)= 0\} } \\
&=&\frac{1}{p^2}\sum_{a\in C_0}\sum_{b\in \gf(q)}\sum_{y\in \gf(p)}\zeta_{p}^{y\tr_{q/p}(\frac{b^2}{4a})+y\beta^{2t}}\sum_{z\in \gf(p)}\chi(zb)\\
&=&\frac{p^{m-2}(p^m-1)}{2}+0+\frac{1}{p^2}\sum_{y\in \gf(p)^*}\chi'(y\beta^{2t})\sum_{a\in C_0}\sum_{b\in \gf(q)}\chi(\frac{yb^2}{4a})\\
& &+\frac{1}{p^2}\sum_{y\in \gf(p)^*}\chi'(y\beta^{2t})\sum_{z\in \gf(p)^*}\sum_{a\in C_0}\sum_{b\in \gf(q)}\chi(\frac{yb^2}{4a}+zb).
\end{eqnarray*}
By Lemmas \ref{lem-charactersum} and \ref{quadGuasssum},
\begin{eqnarray*}
\lefteqn{ \sum_{y\in \gf(p)^*}\chi'(y\beta^{2t})\sum_{a\in C_0}\sum_{b\in \gf(q)}\chi(\frac{yb^2}{4a}) } \\
&=& \frac{(p^m-1)G(\eta,\chi)}{2}\sum_{y\in \gf(p)^*}\chi'(y\beta^{2t})\eta'(y\beta^{2t})\\
&=& \frac{G(\eta,\chi)G(\eta',\chi')(p^m-1)}{2}=\frac{(-1)^{\frac{(p-1)(m+1)}{4}}p^{\frac{m+1}{2}}(p^m-1)}{2},
\end{eqnarray*}
\begin{eqnarray*}
\lefteqn{ \sum_{y\in \gf(p)^*}\chi'(y\beta^{2t})\sum_{z\in \gf(p)^*}\sum_{a\in C_0}\sum_{b\in \gf(q)}\chi(\frac{yb^2}{4a}+zb) } \\
&=&G(\eta,\chi)\sum_{y\in \gf(p)^*}\chi'(y\beta^{2t})\eta'(y\beta^{2t})\sum_{z\in \gf(p)^*}\sum_{a\in C_0}\chi(-\frac{z^2}{y}a)\\
&=&\frac{1}{2}G(\eta,\chi)\sum_{y\in \gf(p)^*}\chi'(y\beta^{2t})\eta'(y\beta^{2t})\sum_{z\in \gf(p)^*}\sum_{a\in \gf(q)^*}\chi(-\frac{z^2}{y}a^2)\\
&=&\frac{1}{2}(-1)^{\frac{p+1}{2}}G(\eta,\chi)^2(p-1)-\frac{1}{2}(p-1)G(\eta,\chi)G(\eta',\chi')\\
&=&-\frac{p^m(p-1)}{2}-\frac{(-1)^{\frac{(p-1)(m+1)}{4}}(p-1)p^{\frac{m+1}{2}}}{2},
\end{eqnarray*}
due to $\eta(a)=1$. Then
\begin{eqnarray*}
\lefteqn{ \sharp \{(a,b)\in C_0\times  \gf(q):\tr_{q/p}(\frac{b^2}{4a})+\beta^{2t}=0 \mbox{ and }\tr_{q/p}(b)= 0\} }\\
&=&\frac{p^{m-1}(p^{m-1}-1)}{2}+\frac{(-1)^{\frac{(p-1)(m+1)}{4}}p^{\frac{m-1}{2}}(p^{m-1}-1)}{2},
\end{eqnarray*} for all $0\leq t\leq \frac{p-3}{2}$.
\item If $\eta(a)=\eta'\left(-\tr_{q/p}(\frac{b^2}{4a})\right)=-1$, then $a\in \alpha C_0$ and $\tr_{q/p}(\frac{b^2}{4a})\in -\beta C_0'$. Let $\tr_{q/p}(\frac{b^2}{4a})+\beta^{2t+1}=0$ for some $0\leq t\leq \frac{p-3}{2}$. Similarly, we can prove that
    \begin{eqnarray*}
\lefteqn{ \sharp \{(a,b)\in C_0\times  \gf(q):\tr_{q/p}(\frac{b^2}{4a})+\beta^{2t+1}=0 \mbox{ and }\tr_{q/p}(b)= 0\} } \\
&=&\frac{p^{m-1}(p^{m-1}-1)}{2}+\frac{(-1)^{\frac{(p-1)(m+1)}{4}}p^{\frac{m-1}{2}}(p^{m-1}-1)}{2},
\end{eqnarray*} for all $0\leq t\leq \frac{p-3}{2}$.
\end{enumerate}
Hence
$$N_5+N_8=\frac{(p-1)\left(p^{m-1}(p^{m-1}-1)+(-1)^{\frac{(p-1)(m+1)}{4}}p^{\frac{m-1}{2}}(p^{m-1}-1)\right)}{2}$$
and
$$N_8=\frac{p^{m-2}(p-1)^{2}(p^{m-1}-1)}{2}.$$

Fourthly, we determine $N_9$. Note that $N_5+N_7+N_8+N_9$  is equal to the number of the solutions $(a,b)\in \gf(q)\times \gf(q)$ of
$$\left\{\begin{array}{l}
a\neq 0,\\
 \tr_{q/p}(\frac{b^2}{4a})\neq 0,\\
 \eta(a)\eta'\left(-\tr_{q/p}(\frac{b^2}{4a})\right)=1.
\end{array} \right.$$
Similarly to the proof of $N_5+N_7$, we can prove that
$$N_5+N_7+N_8+N_9=\frac{(p-1)(p^m-1)(p^{m-1}+(-1)^{\frac{(p-1)(m+1)}{4}}p^{\frac{m-1}{2}})}{2}$$ implying
$$N_9=\frac{p^{m-2}(p-1)^{2}(p^m-p^{m-1}+1)+(-1)^{\frac{(p-1)(m+1)}{4}}p^{\frac{3(m-1)}{2}}(p-1)^{2}}{2}.$$

The determinations of $N_{10},N_{11},N_{12}$ are similar to those of $N_7,N_8,N_9$ and are omitted here. The proof is completed.
\end{proof}

\begin{theorem}\label{th-x2}
Let $m$ be an odd positive integer, $p$ an odd prime  and $q=p^m$. The subfield code $\C_{(x^2,q)}^{(p)}$ is a $[p^m+1,2m,p^{m-1}(p-1)-p^{\frac{m-1}{2}}]$ $p$-ary code, and has weight enumerator
\begin{eqnarray*}
1+(p^{m-1}-1)(p^{m-2}+1)z^{p^{m-1}(p-1)} +\\
p^{m-2}(p-1)(2p^{m-1}+2p-2)z^{p^{m-1}(p-1)+1}+
p^{m-2}(p-1)^2(p^{m-1}-1)z^{p^{m-1}(p-1)+2} +\\
\frac{(p-1)(p^{m-1}-1)(p^{m-2}+(-1)^{\frac{(p-1)(m+1)}{4}}p^{\frac{m-1}{2}})}{2}  z^{p^{m-1}(p-1)-p^{\frac{m-1}{2}}(-1)^{\frac{(p-1)(m+1)}{4}}} +\\
p^{m-2}(p-1)^{2}(p^{m-1}-1)z^{p^{m-1}(p-1)-p^{\frac{m-1}{2}}(-1)^{\frac{(p-1)(m+1)}{4}}+1}+ \\
\frac{p^{m-2}(p-1)^{2}(p^m-p^{m-1}+1)+(-1)^{\frac{(p-1)(m+1)}{4}}p^{\frac{3(m-1)}{2}}(p-1)^{2}}{2}z^{p^{m-1}(p-1)-p^{\frac{m-1}{2}}(-1)^{\frac{(p-1)(m+1)}{4}}+2} +  \\
\frac{(p-1)(p^{m-1}-1)(p^{m-2}+(-1)^{\frac{(p-1)(m+1)+4}{4}}p^{\frac{m-1}{2}})}{2} z^{p^{m-1}(p-1)+p^{\frac{m-1}{2}}(-1)^{\frac{(p-1)(m+1)}{4}}}  + \\
p^{m-2}(p-1)^{2}(p^{m-1}-1)z^{p^{m-1}(p-1)+p^{\frac{m-1}{2}}(-1)^{\frac{(p-1)(m+1)}{4}}+1}  + \\
\frac{p^{m-2}(p-1)^{2}(p^m-p^{m-1}+1)+(-1)^{\frac{(p-1)(m+1)+4}{4}}p^{\frac{3(m-1)}{2}}(p-1)^{2}}{2} z^{p^{m-1}(p-1)+p^{\frac{m-1}{2}}(-1)^{\frac{(p-1)(m+1)}{4}}+2}.
\end{eqnarray*}
Its dual is nearly optimal with respect to the sphere-packing bound, and has parameters $[p^{m}+1,p^{m}+1-2m,3]$.
\end{theorem}

\begin{proof}
Let $\chi$ and $\chi'$ be the canonical additive characters of $\gf(q)$ and $\gf(p)$, respectively. Let $\eta,\eta'$ be the quadratic multiplicative characters of $\gf(q)^*$ and $\gf(p)^*$, respectively. Let ${\bc_{(f,q)}}^{(p)}=\left(\left(\tr_{q/p}(ax^2+bx)\right)_{x\in \gf(q)^*},\tr_{q/p}(a),\tr_{q/p}(b)\right)$ be any codeword in $\C_{(f,q)}^{(p)}$.

Denote by $N_0(a,b,c)=\sharp\{x\in \gf(q):\tr_{q/p}(ax^2+bx)=0\}$. By the orthogonality relation of additive characters,
\begin{eqnarray}\label{eqn-N0}
\nonumber N_0(a,b)&=&\frac{1}{p}\sum_{x\in \gf(q)}\sum_{y\in \gf(p)}\zeta_{p}^{y\tr_{q/p}(ax^2+bx)}\\
\nonumber &=&\frac{q}{p}+\frac{1}{p}\sum_{y\in \gf(p)^*}\sum_{x\in \gf(q)}\chi(yax^2+ybx)\\
&=&p^{m-1}+\frac{1}{p}\Delta(a,b),
\end{eqnarray}
where $\Delta(a,b):=\sum_{y\in \gf(p)^*}\sum_{x\in \gf(q)}\chi(yax^2+ybx)$. We discuss the value of $\Delta(a,b)$ in the following cases.
\begin{enumerate}
\item Let $(a,b,c)=(0,0)$. Then $\Delta(a,b)=q(p-1)$.
\item Let $a=0$ and $b\neq 0$. Then
\begin{eqnarray*}
\Delta(a,b)&=&\sum_{y\in \gf(p)^*}\sum_{x\in \gf(q)}\chi(ybx)=0.
\end{eqnarray*}
\item Let $a\neq 0$. By Lemma \ref{lem-charactersum} we have
\begin{eqnarray*}
\Delta(a,b)&=&\sum_{y\in \gf(p)^*}\chi(-y^2b^2(4ya)^{-1})\eta(ya)G(\eta,\chi)\\
&=&G(\eta,\chi)\eta(a)\sum_{y\in \gf(p)^*}\eta(y)\chi\left(-\frac{b^2}{4a}y\right)\\
&=&G(\eta,\chi)\eta(a)\sum_{y\in \gf(p)^*}\zeta_{p}^{-\tr_{q/p}\left(\frac{b^2}{4a}\right)y}\eta(y)\\
&=&G(\eta,\chi)\eta(a)\sum_{y\in \gf(p)^*}\chi'\left(-\tr_{q/p}\left(\frac{b^2}{4a}\right)y\right)\eta(y).
\end{eqnarray*}
Since $m$ is odd, we have $\eta(y)=\eta'(y)$ for $y\in \gf(p)^*$. Then by Lemma \ref{quadGuasssum} we have
\begin{eqnarray*}
\Delta(a,b)&=&\left\{\begin{array}{ll}
G(\eta,\chi)\eta(a)\sum_{y\in \gf(p)^*}\eta'(y) & \mbox{ if }\tr_{q/p}(\frac{b^2}{4a})=0 \\
\myatop{\mbox{$G(\eta,\chi)\eta(a)\eta'\left(-\tr_{q/p}(\frac{b^2}{4a})\right)\times$}}
{\mbox{$\sum\limits_{y\in \gf(p)^*}\chi'\left(-\tr_{q/p}(\frac{b^2}{4a})y\right)\eta'\left(-\tr_{q/p}(\frac{b^2}{4a})y\right)$}} & \mbox{ if }\tr_{q/p}(\frac{b^2}{4a})\neq 0 \\
\end{array} \right.\\
&=&\left\{\begin{array}{ll}
0 & \mbox{ if }\tr_{q/p}(\frac{b^2}{4a})=0 \\
G(\eta,\chi)G(\eta',\chi')\eta(a)\eta'\left(-\tr_{q/p}(\frac{b^2}{4a}\right) &
\mbox{ if }\tr_{q/p}(\frac{b^2}{4a})\neq 0 \\
\end{array} \right.\\
&=&\left\{\begin{array}{ll}
0 & \mbox{ if }\tr_{q/p}(\frac{b^2}{4a})=0,\\
p^{\frac{m+1}{2}}(-1)^{\frac{(p-1)(m+1)}{4}} &
\mbox{ if }\tr_{q/p}(\frac{b^2}{4a})\neq 0,\ \eta(a)\eta'\left(-\tr_{q/p}(\frac{b^2}{4a})\right)=1,\\
p^{\frac{m+1}{2}}(-1)^{\frac{(p-1)(m+1)+4}{4}} &
\mbox{ if }\tr_{q/p}(\frac{b^2}{4a})\neq 0,\ \eta(a)\eta'\left(-\tr_{q/p}(\frac{b^2}{4a})\right)=-1.
\end{array} \right.\\
\end{eqnarray*}
\end{enumerate}
Equation (\ref{eqn-N0}) and the preceding discussions yield
\begin{eqnarray}\label{N0-oddm}
N_0(a,b)=
\left\{\begin{array}{ll}
p^m & \mbox{ for }(a,b)=(0,0),\\
p^{m-1} & \mbox{ for }a=0,\ b\neq 0,\mbox{ or }a\neq 0,\ \tr_{q/p}(\frac{b^2}{4a})=0,\\
p^{m-1}+p^{\frac{m-1}{2}}(-1)^{\frac{(p-1)(m+1)}{4}} &
\mbox{ if }a\neq 0,\ \tr_{q/p}(\frac{b^2}{4a})\neq 0,\ \eta(a)\eta'\left(-\tr_{q/p}(\frac{b^2}{4a})\right)=1,\\
p^{m-1}+p^{\frac{m-1}{2}}(-1)^{\frac{(p-1)(m+1)+4}{4}} &
\mbox{ if }a\neq 0,\ \tr_{q/p}(\frac{b^2}{4a})\neq 0,\ \eta(a)\eta'\left(-\tr_{q/p}(\frac{b^2}{4a})\right)=-1.
\end{array} \right.
\end{eqnarray}

For any codeword $\bc(a,b)=\left((\tr_{q/p}(ax^2+bx))_{x\in \gf(q)},\tr_{q/p}(a),\tr_{q/p}(b)\right)$, by Equation (\ref{N0-oddm}) we deduce the following:
\begin{enumerate}
\item If $(a,b)=(0,0), \mbox{ then }\wt(\bc(a,b))=0$.
\item $\mbox{ If }\left\{\begin{array}{l}
a=0,\\
b\neq 0,\\
\tr_{q/p}(b)=0,
\end{array} \right.$
$\mbox{ or }\left\{\begin{array}{l}
a\neq 0,\\
 \tr_{q/p}(\frac{b^2}{4a})=0,\\
\tr_{q/p}(a)=0,\\
\tr_{q/p}(b)=0,\\
\end{array} \right.$ then $\wt(\bc(a,b))=p^{m-1}(p-1)$. Its frequency is $p^{m-1}-1+N_1=(p^{m-1}-1)(p^{m-2}+1)$ by Lemma \ref{lem-N1-N4}.
\item $\mbox{ If }\left\{\begin{array}{l}
a=0,\\
b\neq 0,\\
\tr_{q/p}(b)\neq 0,
\end{array} \right.$
$\mbox{ or }\left\{\begin{array}{l}
a\neq 0,\\
 \tr_{q/p}(\frac{b^2}{4a})=0,\\
\tr_{q/p}(a)=0,\\
\tr_{q/p}(b)\neq 0,\\
\end{array} \right.$
$\mbox{ or }\left\{\begin{array}{l}
a\neq 0,\\
 \tr_{q/p}(\frac{b^2}{4a})=0,\\
\tr_{q/p}(a)\neq 0,\\
\tr_{q/p}(b)= 0,\\
\end{array} \right.$
 then $\wt(\bc(a,b))=p^{m-1}(p-1)+1$. Its frequency is $p^{m-1}(p-1)+N_2+N_3=p^{m-2}(p-1)(2p^{m-1}+2p-2)$ by Lemma \ref{lem-N1-N4}.
 \item $\mbox{ If }\left\{\begin{array}{l}
a\neq 0,\\
 \tr_{q/p}(\frac{b^2}{4a})=0,\\
\tr_{q/p}(a)\neq 0,\\
\tr_{q/p}(b)\neq  0,\\
\end{array} \right.$
 then $\wt(\bc(a,b))=p^{m-1}(p-1)+2$. Its frequency is $N_4=p^{m-2}(p-1)^2(p^{m-1}-1)$ by Lemma \ref{lem-N1-N4}.
 \item $\mbox{ If }\left\{\begin{array}{l}
a\neq 0,\\
 \tr_{q/p}(\frac{b^2}{4a})\neq 0,\\
 \eta(a)\eta'\left(-\tr_{q/p}(\frac{b^2}{4a})\right)=1,\\
\tr_{q/p}(a)= 0,\\
\tr_{q/p}(b)= 0,\\
\end{array} \right.$
 then $\wt(\bc(a,b))=p^{m-1}(p-1)-p^{\frac{m-1}{2}}(-1)^{\frac{(p-1)(m+1)}{4}}$. Its frequency is $N_5=\frac{(p-1)(p^{m-1}-1)(p^{m-2}+(-1)^{\frac{(p-1)(m+1)}{4}}p^{\frac{m-1}{2}})}{2}$ by Lemma \ref{lem-N5-N12}.
 \item $\mbox{ If }\left\{\begin{array}{l}
a\neq 0,\\
 \tr_{q/p}(\frac{b^2}{4a})\neq 0,\\
 \eta(a)\eta'\left(-\tr_{q/p}(\frac{b^2}{4a})\right)=1,\\
\tr_{q/p}(a)= 0,\\
\tr_{q/p}(b)\neq  0,\\
\end{array} \right.$
$\mbox{ or }\left\{\begin{array}{l}
a\neq 0,\\
 \tr_{q/p}(\frac{b^2}{4a})\neq 0,\\
 \eta(a)\eta'\left(-\tr_{q/p}(\frac{b^2}{4a})\right)=1,\\
\tr_{q/p}(a)\neq 0,\\
\tr_{q/p}(b)=  0,\\
\end{array} \right.$
 then $\wt(\bc(a,b))=p^{m-1}(p-1)-p^{\frac{m-1}{2}}(-1)^{\frac{(p-1)(m+1)}{4}}+1$. Its frequency is $N_7+N_8=p^{m-2}(p-1)^{2}(p^{m-1}-1)$ by Lemma \ref{lem-N5-N12}.
 \item $\mbox{ If }\left\{\begin{array}{l}
a\neq 0,\\
 \tr_{q/p}(\frac{b^2}{4a})\neq 0,\\
 \eta(a)\eta'\left(-\tr_{q/p}(\frac{b^2}{4a})\right)=1,\\
\tr_{q/p}(a)\neq 0,\\
\tr_{q/p}(b)\neq  0,\\
\end{array} \right.$
 then $\wt(\bc(a,b))=p^{m-1}(p-1)-p^{\frac{m-1}{2}}(-1)^{\frac{(p-1)(m+1)}{4}}+2$. Its frequency is $N_9=\frac{p^{m-2}(p-1)^{2}(p^m-p^{m-1}+1)+(-1)^{\frac{(p-1)(m+1)}{4}}p^{\frac{3(m-1)}{2}}(p-1)^{2}}{2}$ by Lemma \ref{lem-N5-N12}.
  \item $\mbox{ If }\left\{\begin{array}{l}
a\neq 0,\\
 \tr_{q/p}(\frac{b^2}{4a})\neq 0,\\
 \eta(a)\eta'\left(-\tr_{q/p}(\frac{b^2}{4a})\right)=-1,\\
\tr_{q/p}(a)= 0,\\
\tr_{q/p}(b)= 0,\\
\end{array} \right.$
 then $\wt(\bc(a,b))=p^{m-1}(p-1)+p^{\frac{m-1}{2}}(-1)^{\frac{(p-1)(m+1)}{4}}$. Its frequency is $N_6=\frac{(p-1)(p^{m-1}-1)(p^{m-2}+(-1)^{\frac{(p-1)(m+1)+4}{4}}p^{\frac{m-1}{2}})}{2}$ by Lemma \ref{lem-N5-N12}.
 \item $\mbox{ If }\left\{\begin{array}{l}
a\neq 0,\\
 \tr_{q/p}(\frac{b^2}{4a})\neq 0,\\
 \eta(a)\eta'\left(-\tr_{q/p}(\frac{b^2}{4a})\right)=-1,\\
\tr_{q/p}(a)= 0,\\
\tr_{q/p}(b)\neq  0,\\
\end{array} \right.$
$\mbox{ or }\left\{\begin{array}{l}
a\neq 0,\\
 \tr_{q/p}(\frac{b^2}{4a})\neq 0,\\
 \eta(a)\eta'\left(-\tr_{q/p}(\frac{b^2}{4a})\right)=-1,\\
\tr_{q/p}(a)\neq 0,\\
\tr_{q/p}(b)=  0,\\
\end{array} \right.$
 then $\wt(\bc(a,b))=p^{m-1}(p-1)+p^{\frac{m-1}{2}}(-1)^{\frac{(p-1)(m+1)}{4}}+1$. Its frequency is $N_{10}+N_{11}=p^{m-2}(p-1)^{2}(p^{m-1}-1)$ by Lemma \ref{lem-N5-N12}.
 \item $\mbox{ If }\left\{\begin{array}{l}
a\neq 0,\\
 \tr_{q/p}(\frac{b^2}{4a})\neq 0,\\
 \eta(a)\eta'\left(-\tr_{q/p}(\frac{b^2}{4a})\right)=-1,\\
\tr_{q/p}(a)\neq 0,\\
\tr_{q/p}(b)\neq  0,\\
\end{array} \right.$
 then $\wt(\bc(a,b))=p^{m-1}(p-1)+p^{\frac{m-1}{2}}(-1)^{\frac{(p-1)(m+1)}{4}}+2$. Its frequency is $N_{12}=\frac{p^{m-2}(p-1)^{2}(p^m-p^{m-1}+1)+(-1)^{\frac{(p-1)(m+1)+4}{4}}p^{\frac{3(m-1)}{2}}(p-1)^{2}}{2}$.
\end{enumerate}
Then the weight enumerator of $\C_{(f,q)}^{(p)}$ follows.

The dimension of $\C_{(f,q)}^{(p)}$ is $2m$ as $A_0=1$. By Theorem \ref{th-dualdistance}, the minimal distance $d^{(p)\perp}$ of $\C_{(f, q)}^{(p)\perp}$ satisfies $d^{(p)\perp}\geq 3$ as the dual of $\C_{(f, q)}$ has minimal distance 3.  From the weight distribution of $\C_{(f,q)}^{(p)}$ and the first four Pless power moments in \cite[Page 131]{HP}, we can prove that $A_3^{(p)\perp}>0$, where $A_3^{(p)\perp}$ denotes the frequency of the codewords with weight 3 in $\C_{(f, q)}^{(p)\perp}$. Then the parameters of $\C_{(f, q)}^{(p)\perp}$ follow. By the sphere-packing bound, one can deduce that $d^{(p)\perp}\leq 4$. Hence the dual of $\C_{(f, q)}^{(p)}$ is nearly optimal with respect to the sphere-packing bound.
\end{proof}

If $m$ is even, we can similarly prove that the subfield code $\C_{(f,q)}^{(p)}$ has parameters $[p^m+1,2m,(p-1)(p^{m-1}-p^{\frac{m-2}{2}})]$. We omit the proof here.
\begin{example}Let $f(x)=x^{2}$ with $q=p^m$ and $m$ odd.
\begin{enumerate}
\item Let $p=3$ and $m=3$. Then the set $\C_{(f, q)}^{(p)}$ in Theorem \ref{th-x2} is a $[28,6,15]$ ternary code which has the best known parameters, and its dual a $[28,22,3]$ code, while the corresponding best known parameters are $[28,22,4]$ according to the Code Tables at http://www.codetables.de/.
\item Let $p=5$ and $m=3$. Then the set $\C_{(f, q)}^{(p)}$ in Theorem \ref{th-x2} is a $[126,6,95]$  code whose dual is a $[126,120,3]$ code, while the corresponding best known parameters are $[126,6,95]$ and $[126,120,4]$ according to the Code Tables at http://www.codetables.de/.
\end{enumerate}
\end{example}

\subsection{The subfield code $\C_{(f,q)}^{(p)}$ for $f(x)=x^3$ and $p=2$}

Let $f(x)=x^3$ and $p=2$. Then $\gcd(q-1,3-1)=1$ and $\C_{(x^3, q)}$ is a $[q+1, 2, q]$ MDS code  by Lemma \ref{lem-monomial}. By Theorem \ref{thm-tracerepresent},
the binary subfield code of $\C_{(x^3, q)}$ is given by
$$\C_{(x^3,q)}^{(2)}=\left\{{\bc_{(x^3,q)}}^{(2)}=\left(\left(\tr_{q/p}(ax^3+bx)\right)_{x\in \gf(q)^*},\tr_{q/p}(a),\tr_{q/p}(b)\right):\substack{a\in \gf(q)\\b\in \gf(q)}\right\}.$$
In the following, we only investigate the parameters of $\C_{(x^3, q)}^{(2)}$ for odd $m$. For even $m$, the parameters of $\C_{(x^3, q)}^{(2)}$ can be discussed in a similar way.

Let $\chi$ be the canonical additive character of $\gf(q)$. Define a class of exponential sums as
$$S(a,b)=\sum_{x\in \gf(q)}\chi(ax^3+bx),\ a,b\in \gf(q).$$
Since $m$ is odd, $\gcd(3,q-1)=1$. If $a\neq 0$, then there exists exactly one element $c\in \gf(q)^*$ such that $a=c^3$. Hence
$$S(a,b)=\sum_{x\in \gf(q)}\chi(c^3x^3+bx)=\sum_{x\in \gf(q)}\chi(x^3+bc^{-1}x),\ b\in \gf(q),c\in \gf(q^*).$$
As a direct consequence of \cite[Theorem 2]{Carlitz79}, the following lemma can be derived.
\begin{lemma}\label{lem-Sab}
Let $q=2^m$ with $m$ odd and $a,b\in \gf(q)^*$. Then
\begin{eqnarray*}
S(a,b)=\left\{\begin{array}{ll}
0  & \mbox{ if }\tr_{q/2}(bc^{-1})=0,\\
(-1)^{\frac{m^2-1}{8}}2^{\frac{m+1}{2}} & \mbox{ if }bc^{-1}=t^4+t+1\mbox{ and }\tr_{q/2}(t^3+t)=0,\\
-(-1)^{\frac{m^2-1}{8}}2^{\frac{m+1}{2}} & \mbox{ if }bc^{-1}=t^4+t+1\mbox{ and }\tr_{q/2}(t^3+t)=1,\\
\end{array} \right.\\
\end{eqnarray*}
where $a=c^3$. Specially, $S(1,1)=(-1)^{\frac{m^2-1}{8}}2^{\frac{m+1}{2}}$.
\end{lemma}

\begin{theorem}\label{th-x3}
Let $p=2$ and $m\geq 5$ be odd. Then the binary subfield code $\C_{(x^3, q)}^{(2)}$ has parameters $[2^m+1,2m, 2^{m-1}- 2^{(m-1)/2}]$ if $\frac{m^2-1}{8}$ is odd, and $[2^m+1,2m, d^{(2)}\geq2^{m-1}- 2^{(m-1)/2}]$ if $\frac{m^2-1}{8}$ is even.
\end{theorem}
\begin{proof}
Let $\chi$ be the canonical additive character of $\gf(q)$.
Denote $N_0(a,b)=\sharp \{x\in \gf(q):\tr_{q/2}(ax^3+bx)=0\}$. By the orthogonality relation of additive characters, we have
\begin{eqnarray}\label{eqn-x3-1}
\nonumber N_0(a,b)&=&\frac{1}{2}\sum_{z\in \gf(2)}\sum_{x\in \gf(q)}(-1)^{z\tr_{q/2}(ax^3+bx)}\\
\nonumber &=&2^{m-1}+\frac{1}{2}\sum_{x\in \gf(q)}\chi(ax^3+bx)\\
&=&2^{m-1}+\frac{1}{2}S(a,b).
\end{eqnarray}
 We discuss the value of $S(a,b)$ in the  cases below.
\begin{enumerate}
\item Let $a=b=0$. Then $S(a,b)=q$.
\item Let $a=0,b\neq 0$. Then $S(0,b)=\sum_{x\in \gf(q)}\chi(bx)=0$.
\item Let $a\neq 0,b=0$. Then $S(a,0)=\sum_{x\in \gf(q)}\chi(ax^3)$. By Lemma \ref{lem-p-polynomial} we have
\begin{eqnarray*}
S(a,0)^2&=&\sum_{x\in \gf(q)}\chi(ax^3)\sum_{x_1\in \gf(q)}\chi(ax_{1}^{3})\\
&=&\sum_{x\in \gf(q)}\chi(ax^3)\sum_{y\in \gf(q)}\chi(a(x+y)^{3})\\
&=&\sum_{x,y\in \gf(q)}\chi\left(a(x^2+y^2)(x+y)+ax^3\right)\\
&=&\sum_{y\in \gf(q)}\chi(ay^{3})\sum_{x\in \gf(q)}\chi(ayx^2+ay^2x)\\
&=&q+q\sum_{\myatop{y\in \gf(q)^*}{ay(1+ay^3)=0}}\chi(ay^{3})\\
&=&q+q\sum_{ay^3=1}\chi(1)=q-q=0,
\end{eqnarray*}
where we used the variable transformation $x_1=x+y$ in the second equality. Then $S(a,0)=0$.
\item Let $a,b\neq 0$. The value of $S(a,b)$ is given in Lemma \ref{lem-Sab}.
\end{enumerate}
By Equation (\ref{eqn-x3-1}) and the discussions above, we have
\begin{eqnarray}\label{eqn-3}
 N_0(a,b)=\left\{\begin{array}{ll}
 2^m & \mbox{ if }a=b=0,\\
 2^{m-1} & \myatop{\mbox{if $\tr_{q/2}(bc^{-1})=0,a\neq 0,$}}{\mbox{or $a=0,b\neq 0,$}}\\
2^{m-1}+(-1)^{\frac{m^2-1}{8}}2^{\frac{m-1}{2}} & \mbox{ if }bc^{-1}=t^4+t+1\mbox{ and }\tr_{q/2}(t^3+t)=0,\\
2^{m-1}-(-1)^{\frac{m^2-1}{8}}2^{\frac{m-1}{2}} & \mbox{ if }bc^{-1}=t^4+t+1\mbox{ and }\tr_{q/2}(t^3+t)=1,\\
\end{array} \right.
\end{eqnarray}
where $a=c^3$ if $a\neq 0$.

For any codeword $\bc_{(x^3,q)}^{(2)}=\left((\tr_{q/2}(ax^3+bx))_{x\in \gf(q)^*},\tr_{q/2}(a),\tr_{q/2}(b)\right)\in \C_{(x^3,q)}^{(2)}$, by Equation (\ref{eqn-3}) we deduce that
\begin{eqnarray*}
\lefteqn{ \wt(\bc(a,b)) } \\
&=&\left\{\begin{array}{ll}
 0 & \mbox{if }a=b=0,\\
 2^{m-1} & \mbox{if }a=0,\ b\neq 0,\ \tr_{q/2}(b)=0,\\
 2^{m-1}+1 & \mbox{if }a=0,\ b\neq 0,\ \tr_{q/2}(b)\neq0,\\
 2^{m-1} & \mbox{if $\tr_{q/2}(bc^{-1})=0,\ a\neq 0,\ \tr_{q/2}(a)=\tr_{q/2}(b)=0$},\\
 2^{m-1}+1 & \mbox{if $\tr_{q/2}(bc^{-1})=0,\ a\neq 0,\ \tr_{q/2}(a)=1,\ \tr_{q/2}(b)=0$},\\
 2^{m-1}+1 & \mbox{if $\tr_{q/2}(bc^{-1})=0,\ a\neq 0,\ \tr_{q/2}(a)=0,\ \tr_{q/2}(b)=1$},\\
 2^{m-1}+2 & \mbox{if $\tr_{q/2}(bc^{-1})=0,\ a\neq 0,\ \tr_{q/2}(a)=1,\ \tr_{q/2}(b)=1$},\\
 2^{m-1}+(-1)^{\frac{m^2-1}{8}}2^{\frac{m-1}{2}} & \substack{\mbox{if $bc^{-1}=t^4+t+1,\tr_{q/2}(t^3+t)=0,$}\\ \mbox{ and $\tr_{q/2}(a)=\tr_{q/2}(b)=0,$}}\\
2^{m-1}-(-1)^{\frac{m^2-1}{8}}2^{\frac{m-1}{2}} & \substack{\mbox{if $bc^{-1}=t^4+t+1,\tr_{q/2}(t^3+t)=1,$}\\ \mbox{ and $\tr_{q/2}(a)=\tr_{q/2}(b)=0,$}}\\
2^{m-1}+(-1)^{\frac{m^2-1}{8}}2^{\frac{m-1}{2}}+1 & \substack{\mbox{if $bc^{-1}=t^4+t+1, \tr_{q/2}(t^3+t)=0$ and exactly}\\ \mbox{ one of $\tr_{q/2}(a),\tr_{q/2}(b)$ equals 0},}\\
2^{m-1}-(-1)^{\frac{m^2-1}{8}}2^{\frac{m-1}{2}}+1 & \substack{\mbox{if $bc^{-1}=t^4+t+1, \tr_{q/2}(t^3+t)=1$ and exactly}\\ \mbox{ one of $\tr_{q/2}(a),\tr_{q/2}(b)$ equals 0,}}\\
 2^{m-1}+(-1)^{\frac{m^2-1}{8}}2^{\frac{m-1}{2}}+2 & \substack{\mbox{if $bc^{-1}=t^4+t+1,\tr_{q/2}(t^3+t)=0,$}\\ \mbox{ and $\tr_{q/2}(a)=\tr_{q/2}(b)=1,$}}\\
2^{m-1}-(-1)^{\frac{m^2-1}{8}}2^{\frac{m-1}{2}}+2 & \substack{\mbox{if $bc^{-1}=t^4+t+1,\tr_{q/2}(t^3+t)=1,$}\\ \mbox{ and $\tr_{q/2}(a)=\tr_{q/2}(b)=1,$}}\\
\end{array} \right.\\
\end{eqnarray*}
where $a=c^3$ if $a\neq 0$. The dimension is $2m$ as $\wt(\bc(a,b))=0$ if and only if $a=b=0$.  The minimal distance
$$d^{(2)}\geq \min\left\{2^{m-1}+(-1)^{\frac{m^2-1}{8}}2^{\frac{m-1}{2}},2^{m-1}-(-1)^{\frac{m^2-1}{8}}2^{\frac{m-1}{2}}\right\}=2^{m-1}-2^{\frac{m-1}{2}}.$$ Observe that
$$A_{2^{m-1}+(-1)^{\frac{m^2-1}{8}}2^{\frac{m-1}{2}}}=\sharp \left\{(a,b)\in \gf(q)^*\times \gf(q)^*:\substack{\mbox{$bc^{-1}=t^4+t+1,\tr_{q/2}(t^3+t)=0,$}\\ \mbox{ and $\tr_{q/2}(a)=\tr_{q/2}(b)=0$}}\right\}.$$
If $a=b^3$, then $b=c$ and $t=0,1$ implying  $\tr_{q/2}(t^3+t)=0$. Hence
\begin{eqnarray*}A_{2^{m-1}+(-1)^{\frac{m^2-1}{8}}2^{\frac{m-1}{2}}}
&\geq & \sharp \{b\in \gf(q)^*:\tr_{q/2}(b^3)=\tr_{q/2}(b)=0\}\\
&=&\frac{1}{4}\sum_{b\in \gf(q)^*}\sum_{y\in \gf(2)}(-1)^{y\tr_{q/2}(b^3)}\sum_{z\in \gf(2)}(-1)^{z\tr_{q/2}(b)}\\
&=&\frac{q-1}{4}+\frac{1}{4}\sum_{b\in \gf(q)^*}\chi(b^3)+\frac{1}{4}\sum_{b\in \gf(q)^*}\chi(b)+\frac{1}{4}\sum_{b\in \gf(q)^*}\chi(b^3+b)\\
&=&\frac{q-1}{4}+\frac{1}{4}(S(1,0)-1)-\frac{1}{4}+\frac{1}{4}(S(1,1)-1)\\
&=&2^{m-2}-1+(-1)^{\frac{m^2-1}{8}}2^{\frac{m-3}{2}}\\
&>& 0
\end{eqnarray*}
by Lemma \ref{lem-Sab}, where $m\geq 5$. Then the desired conclusion follows.
\end{proof}

In Theorem \ref{th-x3}, we were unable to obtain the minimal distance of $\C_{(x^3, q)}^{(2)}$ if $\frac{m^2-1}{8}$ is even. The weight distribution of $\C_{(x^3, q)}^{(2)}$ is even more difficult to compute.  Our Magma experiments lead to the following conjecture.

\begin{conj}
Let $m \geq 5$ be odd. Then the set $\C_{(x^3, q)}^{(2)}$ in Theorem \ref{th-x3} is a nine-weight code with parameters
$[2^m+1, 2m, 2^{m-1}-2^{(m-1)/2}]$. Its dual has parameters $[2^m+1, 2^m+1-2m, 3]$.
\end{conj}

\begin{example} Let $f(x)=x^{3}$ with $q=2^m$ and $m$ odd.
\begin{enumerate}
\item Let $m=3$. Then the set $\C_{(x^3, q)}^{(2)}$ in Theorem \ref{th-x3} is a $[9,6,2]$ binary code which has the best known parameters according to the Code Tables at http://www.codetables.de/.
\item Let $m=5$. Then the set $\C_{(x^3, q)}^{(2)}$ in Theorem \ref{th-x3} is a $[33,10,12]$  binary code which has the best known parameters according to the Code Tables at http://www.codetables.de/.
\item Let $m=7$. Then the set $\C_{(x^3, q)}^{(2)}$ in Theorem \ref{th-x3} is a $[129,14,56]$  binary code with the best  parameters known  according to the Code Tables at http://www.codetables.de/.
\end{enumerate}
\end{example}

\section{Families of $[2^m+1, 2, 2^m]$ MDS codes from oval polynomials and their subfield codes}

Let $p=2$ and $q=2^m$ throughout this subsection. We first define oval polynomials $f(x)$ on $\gf(q)$, and then
investigate their codes $\C_{(f, q)}$ and $\C_{(f, q)}^{(2)}$.

An oval polynomial $f$ over $\gf(q)$ is a polynomial such that
\begin{enumerate}
\item $f$ is a permutation polynomial of $\gf(q)$ with $\deg(f)<q$ and $f(0)=0$, $f(1)=1$;  and
\item for each $a \in \gf(q)$, $g_a(x):=(f(x+a)+f(a))x^{q-2}$ is also a permutation polynomial
      of $\gf(q)$.
\end{enumerate}

The following is a list of known infinite families of oval polynomials in the literature.

\begin{theorem}\label{thm-knownopolys}
Let $m \geq 2$ be an integer. The following are oval polynomials of $\gf(q)$, where $q=2^m$.
\begin{itemize}
\item The translation polynomial $f(x)=x^{2^h}$, where $\gcd(h, m)=1$.
\item The Segre polynomial $f(x)=x^6$, where $m$ is odd.
\item The Glynn oval polynomial $f(x)=x^{3 \times 2^{(m+1)/2} +4}$, where $m$ is odd.
\item The Glynn oval polynomial $f(x)=x^{ 2^{(m+1)/2} + 2^{(m+1)/4} }$ for $m \equiv 3 \pmod{4}$.
\item The Glynn oval polynomial $f(x)=x^{ 2^{(m+1)/2} + 2^{(3m+1)/4} }$ for $m \equiv 1 \pmod{4}$.
\item The Cherowitzo oval polynomial $f(x)=x^{2^e}+x^{2^e+2}+x^{3 \times 2^e+4},$ where $e=(m+1)/2$ and $m$ is odd.
\item The Payne oval polynomial $f(x)=x^{\frac{2^{m-1}+2}{3}} + x^{2^{m-1}} + x^{\frac{3 \times 2^{m-1}-2}{3}}$,
        where $m$ is odd.
\item The Subiaco polynomial
$$
f_a(x)=((a^2(x^4+x)+a^2(1+a+a^2)(x^3+x^2)) (x^4 + a^2 x^2+1)^{2^m-2}+x^{2^{m-1}},
$$
where $\tr_{q/2}(1/a)=1$ and $d \not\in \gf(4)$ if $m \equiv 2 \bmod{4}$.
\item The Adelaide oval polynomial
$$
f(x)=\frac{T(\beta^m)(x+1)}{T(\beta)} + \frac{T((\beta x + \beta^q)^m)}{T(\beta) (x+T(\beta)x^{2^{m-1}} +1)^{m-1}} + x^{2^{m-1}},
$$
where $m \geq 4$ is even, $\beta \in \gf(q^2) \setminus \{1\}$ with $\beta^{q+1}=1$, $m \equiv \pm (q-1)/3 \pmod{q+1}$,
and $T(x)=x+x^q$.
\end{itemize}
\end{theorem}

The next theorem gives a characterisation of oval polynomials, where the conditions are called
the slope condition, 
and will be needed later.

\begin{theorem}\label{thm-J22220}
$f$ is an oval polynomial over $\gf(q)$ if and only if
\begin{enumerate}
\item $f$ is a permutation of $\gf(q)$; and
\item
$$
\frac{f(x)+f(y)}{x+y} \neq \frac{f(x)+f(z)}{x+z}
$$
for all pairwise-distinct $x, y, z$ in $\gf(q)$.
\end{enumerate}
\end{theorem}

All oval polynomials on $\gf(q)$ can be used to construct $[q+1, 2, q]$ MDS code over $\gf(q)$ in the framework
of this paper. Specifically, we have the following result.

\begin{theorem}
Let $f$ be an oval polynomial over $\gf(q)$. Then $\C_{(f,q)}$ is a $[q+1, 2, q]$ MDS code over $\gf(q)$.
\end{theorem}

\begin{proof}
By definition, $f(a) \neq 0$ for all $a \in \gf(q)^*$. So Condition 1) in Theorem \ref{thm-mainpolycode} is satisfied.
 Condition 2) in Theorem \ref{thm-mainpolycode} follows from Theorem \ref{thm-J22220}. The desired conclusion
 then follows from Theorem \ref{thm-mainpolycode}.
\end{proof}

The subfield code $\C_{(f,q)}^{(2)}$ differs from oval polynomial to oval polynomial. We are able to
settle the parameters of  subfield code $\C_{(f,q)}^{(2)}$ for a few oval polynomials.

By Equation (\ref{matrixG}) and Theorem \ref{th-tracerepresentation}, the trace representation of $\C_{(f,q)}^{(2)}$ is given as
\begin{eqnarray}\label{eqn-trace}
\C_{(f,q)}^{(2)}=\left\{\bc_{(f,q)}^{(2)}=\left(\left(\tr_{q/2}(af(x)+bx)\right)_{x\in \gf(q)^*}, \tr_{q/2}(a), \tr_{q/2}(b)\right):a,b\in \gf(q)\right\}.
\end{eqnarray}

\subsection{The subfield code $\C_{(f,q)}^{(2)}$ for $f(x)=x^2$}
In this subsection, let $f(x)=x^2$ which is an oval polynomial over $\gf(q)$. Then
\begin{eqnarray*}
\C_{(x^2,q)}^{(2)}=\left\{\bc_{(x^2,q)}^{(2)}=\left(\left(\tr_{q/2}(ax^2+bx)\right)_{x\in \gf(q)^*}, \tr_{q/2}(a), \tr_{q/2}(b)\right):a,b\in \gf(q)\right\}
\end{eqnarray*}
by Equation (\ref{eqn-trace}).

\begin{theorem}\label{th-ovalcode1}
Let $m \geq 2$. Then $\C_{(x^2,q)}^{(2)}$ has parameters $[2^m+1, m+1, 2]$ and weight enumerator
$$1+z^2+(2^{m-1}-1)z^{2^{m-1}}+2^mz^{2^{m-1}+1}+(2^{m-1}-1)z^{2^{m-1}+2}.$$
 $(\C_{(x^2,q)}^{(2)})^\perp$ has parameters $[2^m+1, 2^m-m, 3]$ and is dimension-optimal with respect to the sphere-packing bound.
\end{theorem}

\begin{proof}
Let $\chi$ be the canonical additive character of $\gf(q)$.
Denote
$$N_0(a,b)=\sharp \{x\in \gf(q):\tr_{q/2}(ax^2+bx)=0\}.$$ By the orthogonality relation of additive characters and Lemma \ref{lem-charactersum-evenq}, we have
\begin{eqnarray*}\label{eqn-1a}
\nonumber 2N_0(a,b)&=&\sum_{z\in \gf(2)}\sum_{x\in \gf(q)}(-1)^{z\tr_{q/2}(ax^2+bx)}\\
\nonumber&=&q+\sum_{x\in \gf(q)}\chi(ax^2+bx)\\
&=&\left\{\begin{array}{ll}
2q   &   \mbox{ if }a=b^{2},\\
q    &   \mbox{ otherwise. }\\
\end{array} \right.
\end{eqnarray*}
Note that $\tr_{q/2}(b^2)=\tr_{q/2}(b)$. For any codeword $$\bc_{(x^2,q)}^{(2)}=\left((\tr_{q/2}(ax^2+bx))_{x\in \gf(q)},\tr_{q/2}(a),\tr_{q/2}(b)\right)\in \C_{(x^2,q)}^{(2)},$$
Then we have
\begin{eqnarray*}
\wt(\bc(a,b))&=&
\left\{\begin{array}{ll}
q-N_0(a,b)   &   \mbox{ for }a=b^{2},\ \tr_{q/2}(a)=\tr_{q/2}(b)=0 \\
q-N_0(a,b)+2    &   \mbox{ for }a=b^{2},\ \tr_{q/2}(a)=\tr_{q/2}(b)\neq0 \\
q-N_0(a,b)   &   \mbox{ for }a\neq b^{2},\ \tr_{q/2}(a)=\tr_{q/2}(b)=0 \\
q-N_0(a,b)+1 & \myatop{\mbox{ for $a\neq b^{2},\ \tr_{q/2}(a)=0,\ \tr_{q/2}(b)\neq0,$}}{\mbox{ or $a\neq b^{2},\ \tr_{q/2}(a)\neq 0,\ \tr_{q/2}(b)=0$}} \\
q-N_0(a,b)+2    &   \mbox{ for }a\neq b^{2},\ \tr_{q/2}(a)\neq 0,\ \tr_{q/2}(b)\neq0 \\
\end{array} \right.\\
&=&\left\{\begin{array}{ll}
0 &\mbox{ for }a=b^{2},\ \tr_{q/2}(a)=\tr_{q/2}(b)=0,\\
2 &\mbox{ for }a=b^{2},\ \tr_{q/2}(a)=\tr_{q/2}(b)\neq0,\\
2^{m-1}  &\mbox{ for }a\neq b^{2},\ \tr_{q/2}(a)=\tr_{q/2}(b)=0,\\
2^{m-1}+1 &\myatop{\mbox{ for $a\neq b^{2},\ \tr_{q/2}(a)=0,\ \tr_{q/2}(b)\neq0,$}}{\mbox{ or $a\neq b^{2},\ \tr_{q/2}(a)\neq 0,\ \tr_{q/2}(b)=0,$}}\\
2^{m-1}+2 &\mbox{ for }a\neq b^{2},\ \tr_{q/2}(a)\neq 0,\ \tr_{q/2}(b)\neq0.
\end{array} \right.\\
\end{eqnarray*}
Observe that the Hamming weight 0 occurs $2^{m-1}$ times if $(a,b)$ runs through $\gf(q)\times \gf(q)$. Thus every codeword in $\C_{(x^2,q)}^{(2)}$ repeats $2^{m-1}$ times. Based on the discussions above, we easily deduce the weight enumerator of $\C_{(x^2,q)}^{(2)}$.

By Theorem \ref{th-dualdistance}, the minimal distance $d^{(p)\perp}$ of $\C_{(x^2,q)}^{(2)\perp}$ satisfies $d^{(2)\perp}\geq 3$ as the dual of $\C_{(x^2, q)}$ has minimal distance 3.  From the weight distribution of $\C_{(x^2,q)}^{(2)}$ and the first four Pless power moments in \cite[Page 131]{HP}, we can prove that $A_3^{(2)\perp}>0$, where $A_3^{(2)\perp}$ denotes the number of the codewords with weight 3 in $\C_{(x^2,q)}^{(2)\perp}$. Then the parameters of $\C_{(x^2,q)}^{(2)\perp}$  follow.  It is easily verified that
$\C_{(x^2,q)}^{(2)\perp}$ is dimension-optimal with respect to the sphere-packing bound.
\end{proof}

Although the code $\C_{(x^2,q)}^{(2)}$ in Theorem \ref{th-ovalcode1} has bad parameters, its dual code
$(\C_{(x^2,q)}^{(2)})^\perp$ is dimension-optimal. Hence, it is still valuable to study the subfield code $\C_{(x^2,q)}^{(2)}$.

\begin{example}\label{exa-1}
Let $m=2$. Then the set $\C_{(x^2,q)}^{(2)}$ in Theorem \ref{th-ovalcode1} is a $[5,3,2]$ binary linear code and its dual has parameters $[5,2,3]$. Hence $\C_{(x^2,q)}^{(2)}$ is a near MDS code in this case. Both of $\C_{(x^2,q)}^{(2)}$ and its dual has the best known parameters according to the Code Tables at http://www.codetables.de.
\end{example}

\subsection{The subfield code $\C_{(f,q)}^{(2)}$ for $f(x)=x^{2^i+2^j}$ $(i>j\geq 0)$}
Let $f(x)=x^{2^i+2^j}$ $(i>j\geq 0)$. By Theorem \ref{thm-knownopolys}, $f(x)$ is an oval polynomial in the following cases:
\begin{enumerate}
\item $(i,j)=(2,1)$ and $m$ is odd;
\item $(i,j)=((m+1)/2,(m+1)/4)$ and $m \equiv 3 \pmod{4}$;
\item $(i,j)=((3m+1)/4,(m+1)/2 )$ and $m \equiv 1 \pmod{4}$.
\end{enumerate}
 The trace representation of $\C_{(f,q)}^{(2)}$ is given as
\begin{eqnarray*}
\C_{(f,q)}^{(2)}=\left\{\bc_{(f,q)}^{(2)}=\left(\left(\tr_{q/2}(ax^{2^i+2^j}+bx)\right)_{x\in \gf(q)^*},\tr_{q/2}(a),\tr_{q/2}(b)\right):a,b\in \gf(q)\right\}.
\end{eqnarray*}
By \cite[Lemma 5]{WZ}, it is easy to deduce that $\C_{(f,q)}^{(2)}$ is a $[2^m+1, 2m, d^{(2)}\geq 2^{m-1}-2^{(m-1)/2}]$ code if $f(x)$ is one of the above three oval polynomials. However, we were unable to determine its minimal distance and weight distribution. We have the following conjectures according to our Magma experiments.
\begin{conj}
Let $m \geq 5$ be odd. Then $\C_{(x^6,q)}^{(2)}$ has parameters $[2^m+1, 2m, 2^{m-1}-2^{(m-1)/2}]$ and nine nonzero weights.
 $(\C_{(x^6,q)}^{(2)})^\perp$ has parameters $[2^m+1, 2^m-2m+1, 3]$.
\end{conj}

\begin{conj}
Let $m \equiv 3 \pmod{4}\mbox{ with } m\geq 5$  and $f(x)=x^{2^{(m+1)/2}+2^{(m+1)/4}}$. Then $\C_{(f,q)}^{(2)}$ has parameters $[2^m+1, 2m, 2^{m-1}-2^{(m-1)/2}]$ and nine nonzero weights.
 $(\C_{(f,q)}^{(2)})^\perp$ has parameters $[2^m+1, 2^m-2m+1, 3]$.
\end{conj}

\begin{conj}
Let $m \equiv 1 \pmod{4}\mbox{ with } m \geq 5$  and $f(x)=x^{2^{(m+1)/2}+2^{(3m+1)/4}}$. Then $\C_{(f,q)}^{(2)}$ has parameters $[2^m+1, 2m, 2^{m-1}-2^{(m-1)/2}]$ and nine nonzero weights.
 $(\C_{(f,q)}^{(2)})^\perp$ has parameters $[2^m+1, 2^m-2m+1, 3]$.
\end{conj}

By Theorem \ref{thm-knownopolys}, there exist oval polynomials $f$ which are not monomials. It will be very interesting if the parameters of $\C_{(f,q)}^{(2)}$ can be determined with these polynomials.

\section{Summary and concluding remarks}

 In this paper, we first presented a general construction of $[q+1, 2, q]$ MDS codes $\C_{(f,q)}$ over $\gf(q)$ from functions $f$ under certain conditions. Then we  studied the $p$-ary subfield codes of some of the $[q+1, 2,q]$ MDS codes over $\gf(q)$ by selecting some special $f$. These subfield codes and their duals are summarised as follows:
 \begin{enumerate}
 \item A family of three-weight nearly optimal $[p^m+1, m+1, (p-1)p^{m-1}]$ codes with respect to the Griesmer bound whose duals have parameters $[p^m+1, p^m-m, 3]$ and are dimension-optimal with respect to the sphere-packing bound  for $m\geq2$ and any prime $p$ (see Theorem \ref{th-f=1}).

 \item A family of eight-weight $[p^{2l}+1,3l,p^{l-1}(p^{l+1}-p^{l}-1)]$ codes whose duals have parameters
 $[p^{2l}+1,p^{2l}+1-3l,3]$  and are nearly optimal with respect to
 the sphere-packing bound for $l\geq 2$ and  any prime $p$  (see Theorem \ref{th-code1}).

 \item A family of nine-weight $[p^m+1,2m,p^{m-1}(p-1)-p^{\frac{m-1}{2}}]$ codes whose duals have parameters $[p^{m}+1,p^{m}+1-2m,3]$  and are nearly optimal with respect to the sphere-packing bound for odd $m\geq 3$ and odd prime $p$ (see Theorem \ref{th-x2}).

 \item A family of binary $[2^m+1,2m, 2^{m-1}- 2^{(m-1)/2}]$ code for odd $\frac{m^2-1}{8}$ and odd $m\geq 3$, and $[2^m+1,2m, d^{(2)}\geq2^{m-1}- 2^{(m-1)/2}]$ code for even $\frac{m^2-1}{8}$ and odd $m\geq 3$ (see Theorem \ref{th-x3}).

 \item A family of four-weight binary $[2^m+1, m+1, 2]$ codes whose duals have parameters $[2^m+1, 2^m-m, 3]$
 and are dimension-optimal  with respect to the sphere-packing bound  for $m\geq 2$ (see Theorem \ref{th-ovalcode1}).
 \end{enumerate}
Two families of these codes over $\gf(p)$ are dimension-optimal and some families are nearly optimal.  Examples in this paper showed that these
codes over $\gf(p)$ are optimal in some cases.

As pointed out earlier, it is even very hard to determine the parameters of the subfield code of the Simplex code which is a one-weight code over $\gf(p^m)$ and very simple. It is in general very difficult to determine the parameters of subfield codes of linear codes. In this paper, we presented several conjectures about the parameters of some subfield codes.  The reader is cordially invited to settle them.

 Finally, we point out that the subfield codes presented in this paper have various parameters and weight distributions,
 though all of them are constructed from $[q+1, 2, q]$ MDS codes over $\gf(q)$. Although all $[q+1, 2, q]$ MDS codes over
 $\gf(q)$ are monomially equivalent and may not be interesting in many senses, they are very attractive for constructing
 very good linear codes over small fields.  A contribution of this paper is the proof of the fact that $[q+1, 2, q]$ MDS codes over  $\gf(q)$ are very useful and interesting in coding theory.

%

\end{document}